%% file: ZXLS-journal.tex
\documentclass[onecolumn,a4paper,accepted=2019-08-18]{quantumarticle}
\pdfoutput=1
\usepackage{amssymb,amsmath,amsthm,tikz,mathtools,syntax,stmaryrd}
\usepackage[numbers]{natbib}
\theoremstyle{plain}
\newtheorem{theorem}{Theorem}
\newtheorem{lemma}[theorem]{Lemma}
\newtheorem{corollary}[theorem]{Corollary}
\theoremstyle{definition}

\usetikzlibrary{shapes,arrows}

\newcommand{\h}{\mathcal{H}}

\newcommand{\xor}{\leavevmode\hbox{\footnotesize{XOR }}}
\newcommand{\nott}{\leavevmode\hbox{\footnotesize{NOT }}}
\newcommand{\cnot}{\leavevmode\hbox{\footnotesize{CNOT }}}
\newcommand{\ket}[1]{\ensuremath{\left|#1\right\rangle}}
\newcommand{\bra}[1]{\ensuremath{\left\langle #1\right|}}

\newcommand{\ssur}[1]{^{\smash{(#1)}}}

\title{The ZX calculus is a language for surface code lattice surgery}
\date{2 September, 2019}
\author{Niel de Beaudrap}
\affiliation{Department of Computer Science, University of Oxford, Parks Road, Oxford, OX1 3QD}
\email{niel.debeaudrap@cs.ox.ac.uk}
\author{Dominic Horsman}
\affiliation{Department of Physics, Durham University, South Road, Durham, DH1 1LE\\
Department of Computer Science, University of Oxford, Parks Road, Oxford, OX1 3QD}
\email{dom.horsman@gmail.com}

\begin{document}
\maketitle

\begin{abstract}
A leading choice of error correction for scalable quantum computing is the surface code with lattice surgery.
The basic lattice surgery operations, the merging and splitting of logical qubits, act non-unitarily on the logical states and are not easily captured by standard circuit notation.
 This raises the question of how best to design, verify, and optimise protocols that use lattice surgery,
in particular in
architectures with complex
resource management issues. 
In this paper we demonstrate that the
operations of the ZX calculus --- a form of quantum diagrammatic
reasoning based on bialgebras --- match exactly the operations
of lattice surgery. Red and green ``spider'' nodes match rough and
smooth merges and splits, and follow the axioms of a dagger special
associative Frobenius algebra. Some lattice surgery operations
require non-trivial correction operations, which are captured
natively in the use of the ZX calculus in the form of ensembles of
diagrams. We give a first taste of the power of the calculus as a
language for lattice surgery by considering two operations (T gates
and producing a \cnot) and show how ZX diagram re-write rules give
lattice surgery procedures for these operations that are novel,
efficient, and highly configurable.

\end{abstract}

\section{Introduction}
%\vspace{-.25em}

With the development of (small-scale, noisy) devices, quantum computing is in the midst of moving from concept to mature technology \cite{IBM,google}.
As we try to determine how best to make quantum technology scalable, it becomes apparent that new theoretical tools would be useful, both to organise the components of quantum computers, and to reason effectively about the way that these systems are constructed.

A key element of any such large-scale fault-tolerant architecture is quantum error correction \cite{van2013blueprint,terhal2015}.
Currently the technique of choice for flexibility and efficient use of resources is the surface code with lattice surgery \cite{bravyikitaev,melattice,fowler2018low}.
The surface code, in this case, encodes a single logical qubit in the state of many entangled qubits in a single planar surface (a quantum memory). Error detection and correction are performed by the repeated measurement of stabilizers across the surface, to track how noise affects the encoded data.
Lattice surgery operations are performed by splitting surfaces into two or more (using measurements), and by merging two or more surfaces together (by further stabilizer measurements).

An intriguing feature of lattice surgery operations is that they necessarily involve quantum transformations of memories by CPTP maps which are not unitary transformations.
Two-qubit operations in the previously-standard circuit model, such as a \cnot, may be realised deterministically in lattice surgery, using classical processing similar to that involved in standard teleportation.
However, the basic split and merge operations themselves do not rest easily with the standard unitary circuit model. 
The absence of a notation or language specifically for lattice surgery has made it hard efficiently to design, verify, and compile {novel} surgery patterns.
It has also lead to the occasional misrepresentation of lattice surgery as a method for merely simulating unitary gates, rather than a set of computational primitives in its own right. %humour me, I've wanted to say this since 2013
Previous methods for representing lattice surgery include as square patches \cite{herrnoridevitt}, moving squares around in a `game' style \cite{litinski2019game}, and as 3D space-time figures in CAD software \cite{gidney2018efficient}. 
In all cases, these representations get visually and technically unwieldy very quickly, making design, optimisation, and particularly verification (that a lattice surgery procedure is doing what we think) challenging.

In this paper we link the operations of lattice surgery directly to a pre-existing diagrammatic language for quantum computing: the ZX calculus. 
Developed over the last ten years, the ZX calculus (sometimes referred to simply as ``ZX'') is an abstract graphical language for tensor networks that is complete for quantum mechanics \cite{zxbook,HNW}.
We show in this paper how the lattice surgery primitives of split and merge --- including their non-deterministic byproduct operations --- are precisely captured by some of the simplest diagrams of this calculus. 
In order to do this, we also give a more fine-grained analysis of lattice surgery operations and corrections than {appears in the pre-existing literature}.
The ZX calculus is therefore a language for lattice surgery, meaning that results of work on the calculus can be imported directly for use with lattice surgery.

Unlike circuit notation, the ZX calculus is indeed a \emph{calculus}: a formal language, with meaning-preserving rules for how those diagrams may be transformed, without the need to transcribe those operations as exponentially large matrices (even for full pure-state QM, that cannot be efficiently simulated using e.g. stabilizer notation \cite{gottesman1996class}).
By such transformations of these diagrams --- corresponding to different sets of lattice surgery procedures that implement the same operation --- new protocols may thus be discovered.
We give examples of re-writing diagrams for \cnot and T-gate operations to demonstrate this technique.
In the process, we discover six novel \cnot implementations in lattice surgery, and significantly reduce the operational overhead of the T-gate implementation.

The use of the ZX calculus for surface codes with lattice surgery makes the manipulation of error corrected operations visually intuitive, and capable of verification at large scales. As these operations will likely form the basic operations that a fault-tolerant device will use at the logical level, ZX therefore becomes the natural language and logic for programming large-scale quantum computing technologies.

%\vspace{-.5em}
\section{The planar surface code}
%\vspace{-.25em}

The planar surface code uses a 2D lattice segment, in which every edge of the lattice (between neighbouring vertices) is associated with a qubit~\cite{bravyikitaev,freedman2001projective,dennis2002topological,terhal2015}.
These physical qubits encode a single logical qubit's worth of information in their {(generally highly-entangled)} joint state. For a lattice of $n$ qubits there are $n-1$ 
\emph{stabilizers} (operators for which the state of the lattice is a {positive {$+1$}} eigenstate), fixing a subspace which encodes a single degree of freedom (the logical qubit) across the physical qubits. 

Figure \ref{lattice}(a) shows a distance-3 planar surface code lattice (the distance of a code is the weight of the smallest non-trivial error which cannot be detected by the code). The black circles correspond to physical qubits; everything else in the diagram is an aid to the eye. A planar code lattice has two types of boundary: rough (here, left and right) and smooth (here, top and bottom). 
The stabilizers are defined in terms of {local relationships forming} ``plaquettes''. In this paper we use the convention that the term ``plaquette'' refers to any set of qubits, {either} on the boundary of a face or surrounding a vertex of the lattice.
Specifically, the operators ${Z\otimes Z\otimes Z \otimes Z}$ around each face (such as the blue-shaded set of four qubits), and ${X\otimes X\otimes X \otimes X}$ around each vertex (such as the brown shaded region).
For plaquettes on the surface boundary, the corresponding {stabilisers are} $Z \otimes Z \otimes Z$ or $X \otimes X \otimes X$.
Figure \ref{lattice}(a) has 13 physical qubits, and 12 plaquettes/stabilizers. 
The {codespace is the} simultaneous {$+1$} eigenstate of all of the stabilisers: taken as multi-qubit observables, these operators can be measured without changing the {encoded ``lattice state''}. 

The information of the logical qubit is accessed through \emph{logical operators}, {as shown in {Figure~\ref{lattice}(b)}. 
Any {horizontal} chain of $Z$ operators starting and finishing on (separate) rough boundaries is a logical $Z_L$ operator, and any {vertical} chain of $X$ operators between smooth boundaries is a logical $X_L$ operator.}
These operations commute with all of the stabilizers, {and} so preserve the code.
{The significance of the $Z_L$ and $X_L$ operators} is two-fold. Firstly, they {realise} operations on the encoded logical qubit. {For instance, a bit-flip of the logical qubit is accomplished via} a sequence of physical $X$ operations on qubits forming a chain between smooth boundaries {(the $X_L$ operator)}. The logical operators act secondly {as measurable observables of the encoded data itself, encoding specific logical states}. For instance, if a chain of qubits between rough boundaries are measured out in the $Z$ basis, then the {product of the measurement results $\pm 1$} {is the eigenvalue} of the logical $Z_L$. If the result is {$+1$} then the logical state was {$\ket{0}_L$}, and if {$-1$} then it was {$\ket{1}_L$}.

 {Regular measurements of the stabilizers allow for error detection and correction by detecting changes to their eigenvalues.
A Pauli {$X$ or $Z$} error on a physical qubit flips {the eigenvalue of the state to $-1$, for the stabilisers of the opposite type which act on that qubit ($Y$ errors affect both types of stabilisers)}.
A ``string'' of identical Pauli errors {connecting qubits by horizontal and vertical steps} will only flip stabilizers at the endpoints of the string. 
In practise, we avoid explicitly correcting errors (which adds noise owing} to imperfections of the correcting gates), and instead keep track of the accumulated errors classically.
This information is referred to as a ``Pauli frame''~\cite{knill2005pauli}, and plays the role of a reference frame which indicates how information is stored in the error-affected memory.
We refer to the usual reference frame of the surface code, in which the encoded state is a $+1$ eigenstate of every ``plaquette operator'' (i.e.~the state is literally stabilised by those operators) as the \emph{positive (Pauli) frame}.
Pauli frames  require updating not only owing to error, but also after some lattice surgery operations, {described below}.

\begin{figure}
	\centering
	\begin{align*}
      \begin{gathered}%
      \includegraphics[height=5cm]{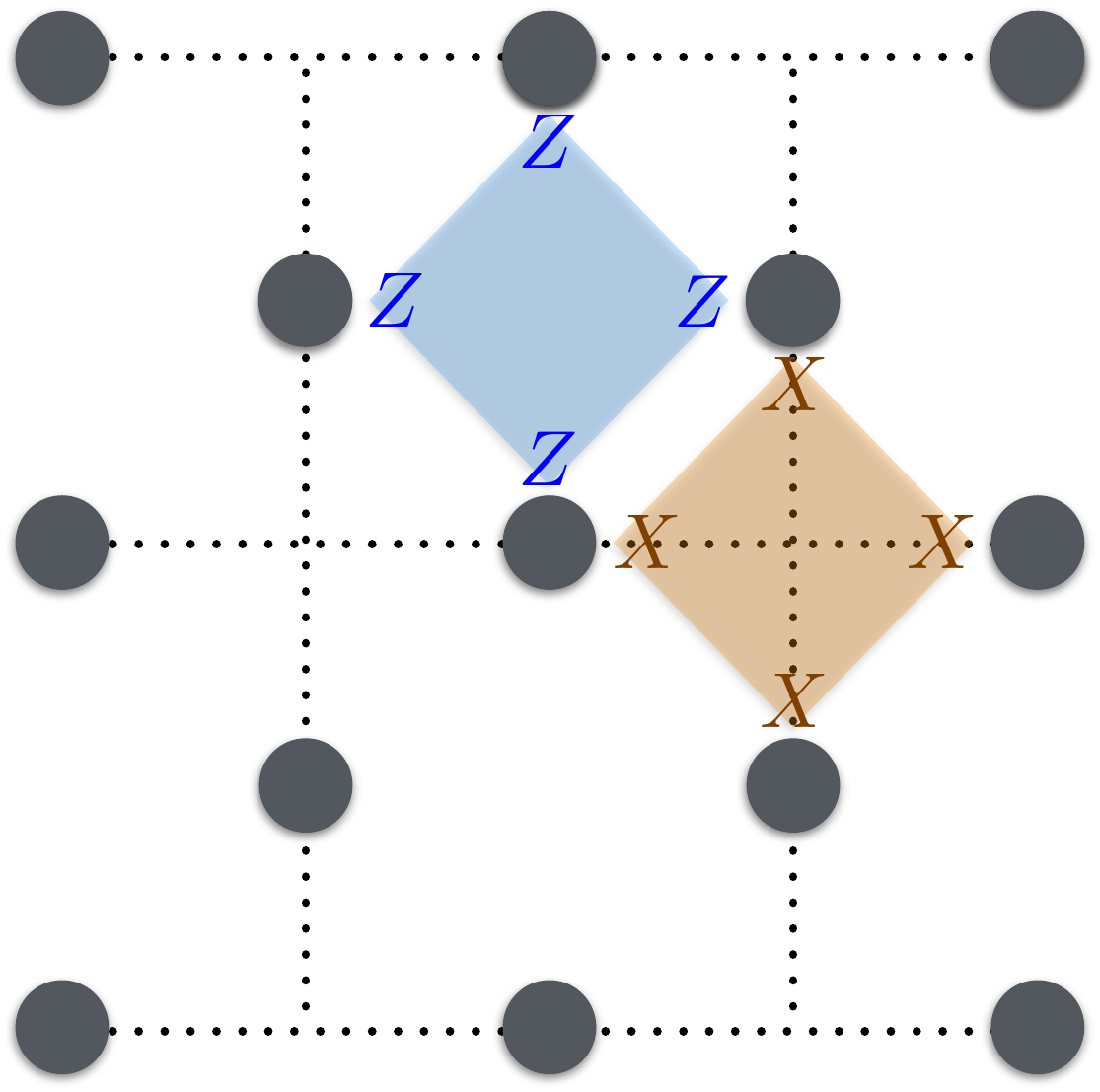}
      	 \\[-3ex]
      	\textup{(a)}
      \end{gathered}
      &&
      \begin{gathered}%
      \includegraphics[height=5cm]{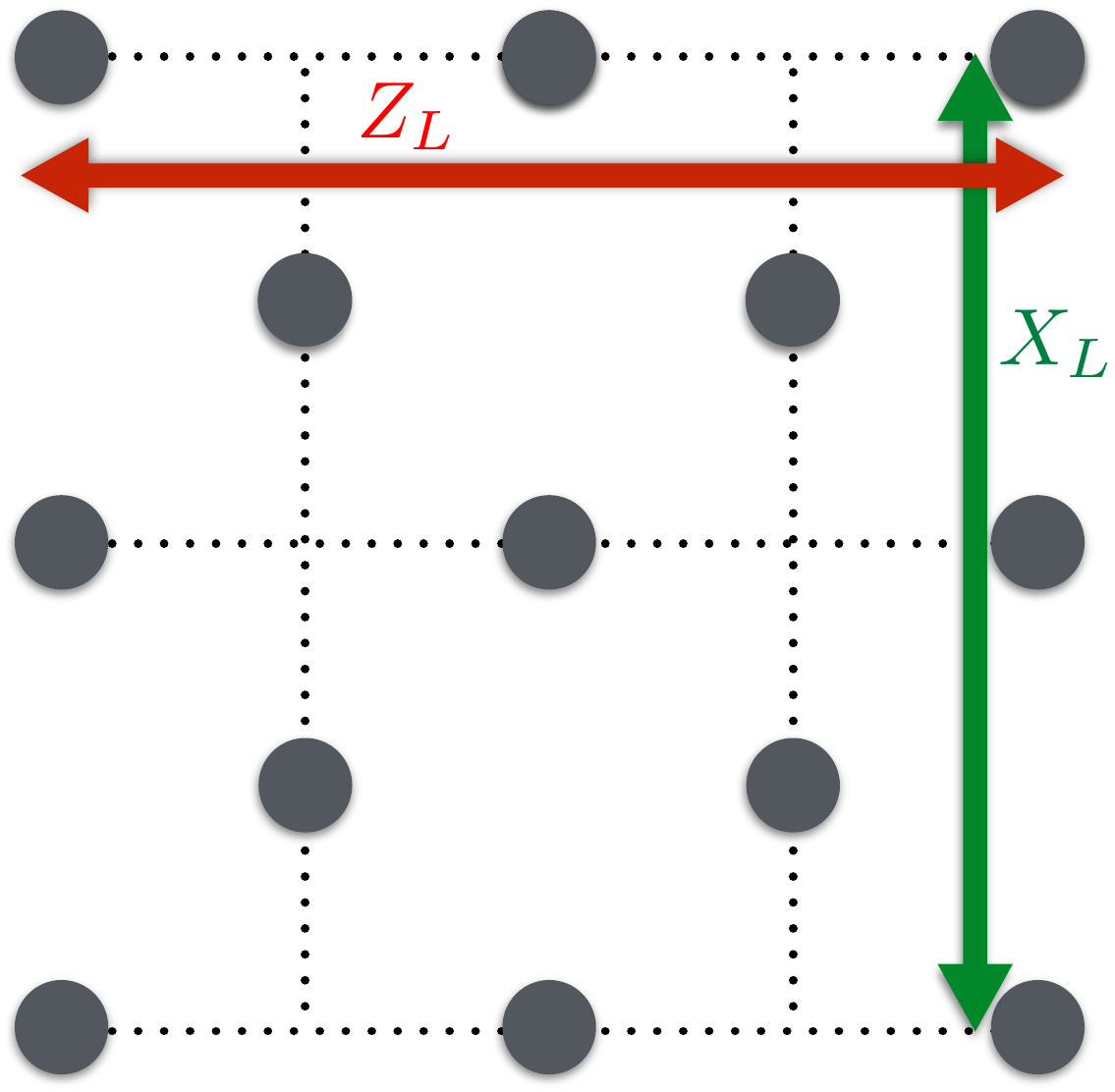}
      	\\[-3ex]
      	\textup{(b)}  
      \end{gathered}
	\end{align*}~\\[-4ex]
	\caption{A distance 3 planar surface code. (a)~Stabilizers on face and vertex plaquettes. (b)~Support of the logical operators: {a $Z_L$ operator  may be realised by a sequence of $Z$ operators on any horizontal row of qubits connecting the two ``rough'' boundaries, and similarly an $X_L$ operator my be realised by a sequence of $X$ operators on any vertical column of qubits connecting the ``smooth'' boundaries.~
    }}
	\label{lattice}
	%\vspace*{1ex}
\end{figure}

%\vspace{-.5em}
\section{{Lattice surgery}}
%\vspace{-.25em}

Lattice surgery is a method of combining or decomposing planar {surface-coded memories,} to perform multi-qubit operations on the encoded information\footnote{The other method for planar codes is transversal \cite{dennis2002topological}, requiring more operations and/or connectivity. Other procedures for surface codes involve introducing defects into the lattice, and braiding \cite{raussendorf2007,fowler2009high} or twisting them \cite{bombin2010topological}, or by code deformation \cite{bombin2009}.} \cite{melattice}. There are two types of surgery operation, \textbf{split} and \textbf{merge}, which can be either ``rough'' or ``smooth''.
These operations all change the number of memories present, necessarily introducing discontinuities in the {Pauli} frames of these memories.
By accounting appropriately for these changes in the reference frames, we may regard the effect of ``split'' operations on the encoded state as a single CPTP map from density operators on $\h \cong \mathbb{C}^2$ to operators on $\h \otimes \h$; and ``merge'' operations as taking a density operator on $\h \otimes \h$ and producing one on $\h$.

%\vspace{-.25em}
\subsection{Splitting}
\label{sec:split}
%\vspace{-.15em}

Figure~\ref{split}(a) {illustrates} a \textbf{rough split}. For {greater clarity}, the memory is shown with a greater width than height, {though this is not required}. {We perform $Z$ measurements on} the intermediate purple qubits {(shown crossed out)}.
The result is {two new planar surface memories}, each encoding {a qubit, with} a boundary where the column of measured qubits used to be.
Most of the stabilizers in the new ``daughter'' memories will be the same operators (and have the same eigenvalues) as in the original ``mother'' memory.
However, along the boundary the daughter memories will have modified stabilisers of weight $3$ rather than $4$,
and some of these may end up in the $-1$ eigenstate, requiring correction to complete the logical operation.
\begin{figure}
	\centering
	\begin{align*}
	  \begin{gathered}%
        \includegraphics[height=3cm]{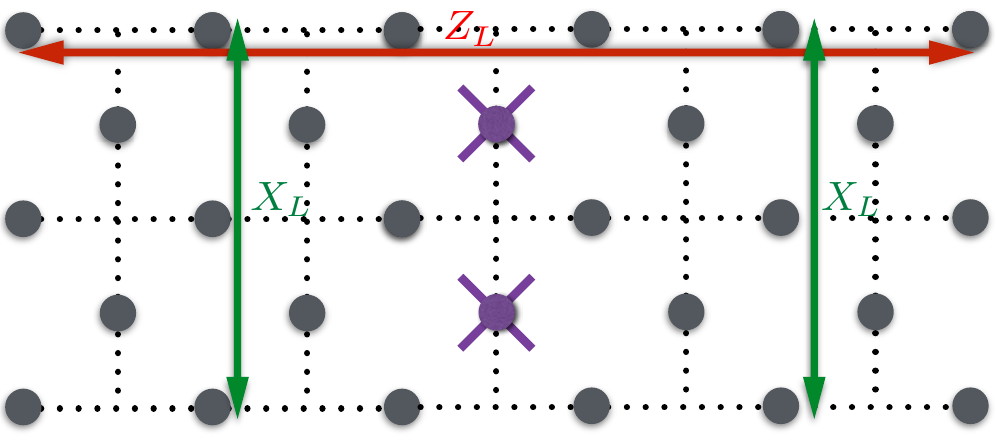}
      \\
        \textup{(a)}
	  \end{gathered}
	  &&\qquad
	  \begin{gathered}%
        \includegraphics[height=3cm]{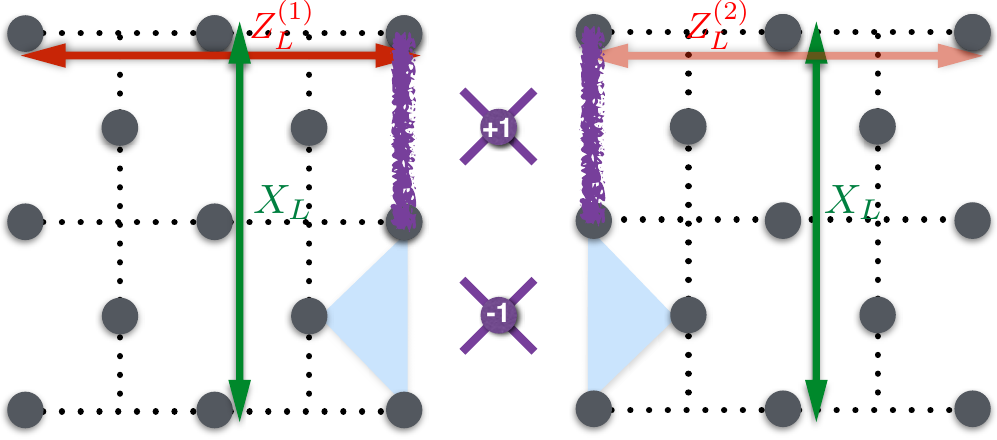}
	  \\
        \textup{(b)}
	  \end{gathered}
	\end{align*}~\\[-4ex]
	\caption{A rough split. (a) Purple qubits are measured out in the $Z$ basis and determine the Pauli frame of the two new surfaces. The split copies $X_L$, and distributes $Z_L$ across the two resulting surfaces. (b)~Correction of a $-1$ measurement result on the lower qubit: a chain of X operations (purple line) joining the errored syndromes (blue shading) to a boundary. }
	\label{split}
	%\vspace*{-1ex}
\end{figure}

Let $X_L\ssur{0}$ represent a logical $X$ operator on the mother memory, and $X_L\ssur{1}$ and $X_L\ssur{2}$ represent logical $X$ operators on the left and right daughters (similarly for $Z_L\ssur{0}$, $Z_L\ssur{1}$, and $Z_L\ssur{2}$).
Note that $X_L\ssur{1}$ and $X_L\ssur{2}$ both commute with the measurements of the split operation, so 
{the expected value of either $X_L\ssur{1}$ or $X_L\ssur{2}$ is equal to $X_L\ssur{0}$ before the split:}
$X_L\ssur{0} \equiv X_L\ssur{1} \equiv X_L\ssur{2}$.
The {$Z_L\ssur{0}$} operator, however, is distributed across the two new surfaces. As $Z_L\ssur{0}$ commutes with the split procedure and yields $Z$ operators across both memories, it decomposes as $Z_L\ssur{0} \equiv Z_L\ssur{1} Z_L\ssur{2}$.
{In particular, if} subsequently measured, the outcome of measuring $Z_L\ssur{1} Z_L\ssur{2}$ is the same as measuring $Z_L\ssur{0}$ before the split: $\langle Z_L\ssur{1} Z_L\ssur{2} \rangle = \langle Z_L\ssur{0} \rangle$.

{If all measurement outcomes of the split are $+1$, then no corrections are applied. However, if one or more measurements give $-1$, as in {Figure~\ref{split}(b)}, this produces matching pairs of $Z$-plaquette stabilizers on the boundaries of each daughter memory that are in the $-1$ eigenstate (shaded blue). {We adapt the Pauli frames (simulating a correction) by} a chain of physical $X$ operations linking to the top boundary (purple chain). The lower boundary could also be used; the difference would be equivalent to a logical {$X_L\ssur{i}$} operation. However, the split has copied $X_L\ssur{0}$, so the daughters are in an eigenstate of $X_L\ssur{1} \otimes X_L\ssur{2}$. As long as the same boundary is chosen for both, the correction strategies differ only by the stabilizer operation $X_L\ssur{1} \otimes X_L\ssur{2}$ and {therefore are} equivalent.}
In this way, we may accommodate the changes in Pauli frame introduced by the measurements, and represent the effect of the split operation on the logical state as a unitary embedding on density operators 
$  \mathsf S_R(\rho) \,=\, U_{R}^{\phantom\dagger}\,\rho\,U_{R}^\dagger$\,,
where $
U_{R} \,=\, \ket{\texttt{++}}\!\!\bra{\texttt{+}} \,+\, \ket{\texttt{--}}\!\!\bra{\texttt{-}}$.

A \textbf{smooth split} performs the corresponding operation {while interchanging $Z$ and $X$, {and }the horizontal and vertical axes.
Performing $X$ measurements  along a row which does not quite reach the two rough boundaries, we obtain two memories
where {(by adapting the Pauli frames appropriately)} the observables on the original and new memories satisfy $Z_L\ssur{0} \equiv Z_L\ssur{1} \equiv Z_L\ssur{2}$ and $X_L\ssur{0} \equiv X_L\ssur{1} X_L\ssur{2}$.}
The effect on the logical state is then another embedding on density operators 
$  \mathsf S_S(\rho) \,=\, U_{S}^{\phantom\dagger}\,\rho\,U_{S}^\dagger$\,,
where $U_{S} \,=\, \ket{00}\!\!\bra{0} \,+\, \ket{11}\!\!\bra{1}$.

%\vspace{-.25em}
\subsection{Merging}
%\vspace{-.15em}

A \textbf{rough merge} ({illustrated} in Figure \ref{merge}) joins two ``parent'' memories along their rough edges. An intermediate column of qubits (initially here in the $\ket{\texttt{+}}$ state, although the $\ket{\texttt{0}}$ state can also be used) is added, and {$X$-plaquette operators are} measured across this join.
The result is a single ``child'' memory, whose Pauli frame is the union of the frames of the parent surfaces, corrected for the outcomes of the plaquette measurements across the join.

Let us denote the logical $X$ operators on either parent by $X_L\ssur{1}$ and $X_L\ssur{2}$, and the logical $X$ on the child memory as $X_L\ssur{3}$ (similarly for $Z_L\ssur{i}$). Measuring the $X$-plaquette operators across the join and taking the product of the outcomes is equivalent to measuring the two columns of $X$ operators on either side of the divide. These are the {$X_L\ssur{i}$} logical operators. The action of the rough merge {therefore realises a ${X_L\ssur{1} \otimes X_L\ssur{2}}$ measurement.
This removes} a degree of freedom, destroying any information carried by expectation values {of} the observables $Z_L\ssur{1}$ or $Z_L\ssur{2}$ alone. 
However, the new {$Z_L\ssur{3}$} operator has support on a chain of $Z$ operators the entire width of the surface: $Z_L\ssur{3} = {Z_L\ssur{1}\otimes Z_L\ssur{2}}$. This commutes with ${X_L\ssur{1} \otimes X_L\ssur{2}}$, and so is unchanged by the merge. The new basis states (the positive and negative eigenstates of {$Z_L\ssur{3}$}) arise from the previous ones, $\ket{00}_{12}, \ket{11}_{12} \mapsto \ket{0}_3$ and $\ket{01}_{12}, \ket{10}_{12} \mapsto \ket{1}_3$: this effectively computes the \xor of the labels of the standard basis states (and similarly for superpositions of these basis states).
\begin{figure}
	\centering
	\includegraphics[width=8cm]{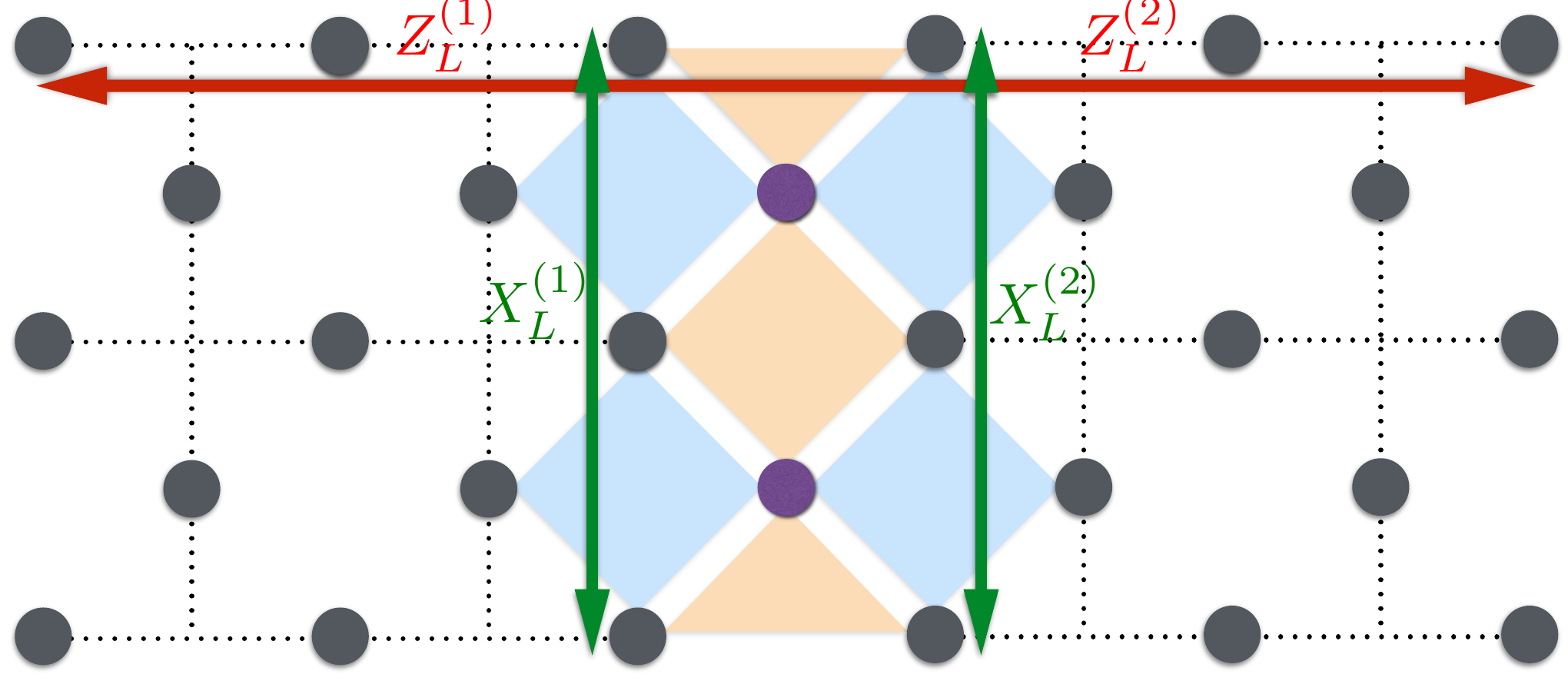}
	%\vspace*{-1ex}
	\caption{A rough merge. Purple qubits are initialised in {$\ket{+}$}. Measuring $X$ plaquette operators (orange) across the join realises a $X_L\ssur{1} \otimes X_L\ssur{2}$ measurement.}
	\label{merge}
	%\vspace*{-1ex}
\end{figure}

The action of the rough merge on the $X_L\ssur{i}$ operators is more subtle.
Consider the case where both parent surfaces are in the positive frame.
Measuring $X_L\ssur{1} X_L\ssur{2}$ by performing the merge tells us whether the parent $X_L$ operators are the same ($+1$) or different ($-1$): {we call this \emph{the outcome} of the merge}.
In the ``positive branch'' ($+1$ outcome), all possible $X_L$ operators on the child surface are identical: $X_L\ssur{1} \equiv X_L\ssur{2} \equiv X_L\ssur{3}$.
If, however, the $X_L\ssur{1} X_L\ssur{2}$ measurement outcome is $-1$ {(the ``negative branch'')}, we have $X_L\ssur{1} \equiv -X_L\ssur{2}$. It is as if one half of the child memory is subject to a string of $Z$ errors. The choice {to correct} from the join either to the left or to the right boundary represents a choice of what logic the merge implements. Either $X_L\ssur{3} \equiv X_L\ssur{1} \equiv -X_L\ssur{2}$, or $X_L\ssur{3} \equiv X_L\ssur{2} \equiv -X_L\ssur{1}$, depending on whether we adapt the Pauli frame of the child memory using a chain of $Z$ operations to the right or to the left respectively. These choices differ by a $Z_L\ssur{3}$ operation, {corresponding to} the difference in reference frames.

Described as a CPTP map, the logical transformation of the rough merge is a map from two-qubit density operators to one-qubit density operators $
  \mathsf M_R(\rho) \,=\, K_{0,R}^{\phantom\dagger} \,\rho\, K_{0,R}^\dagger + K_{1,R}^{\phantom\dagger} \,\rho\, K_{1,R}^\dagger  $. Each Kraus operator $K_{b,R}$ is the transformation in the case of a \smash{$M = (-1)^b$} measurement outcome {of $X_L\ssur{1}X_L\ssur{2}$.}
  If we represent the classical outcome $M = \pm 1$ itself by a decohered quantum system, the operation is a channel
\begin{equation}
  \tilde{\mathsf M}_R(\rho) \,=\,  \Bigl( K_{0,R}^{\phantom\dagger} \,\rho\, K_{0,R}^\dagger \otimes \ket{\texttt{+1}}\!\!\bra{\texttt{+1}}_M\Bigr) \,+\, \Bigl(K_{1,R}^{\phantom\dagger} \,\rho\, K_{1,R}^\dagger \otimes \ket{\texttt{-1}}\!\!\bra{\texttt{-1}}_{M}\Bigr)
\end{equation}
where the $M$ system heralds the effect on the merged parent memories.

The form of the Kraus operators $K_{0,R}$ and $K_{1,R}$ depends on the choice of reference frame for the child surface.
In both cases, we have for the outcome $M = +1$ the Kraus operator
$  K_{0,R} \;=\; \ket{\texttt{+}}\!\!\bra{\texttt{++}} \;+\; \ket{\texttt{-}}\!\!\bra{\texttt{--}} $.
Depending on whether one adapts the Pauli frame {of} the child on the left or the right for the outcome $M = -1$, we respectively obtain the Kraus operator
$K_{1,R} \;=\; \ket{\texttt{+}}\!\!\bra{\texttt{-+}} \;+\; \ket{\texttt{-}}\!\!\bra{\texttt{+-}}$
or $ K_{1,R} \;=\; \ket{\texttt{+}}\!\!\bra{\texttt{+-}} \;+\; \ket{\texttt{-}}\!\!\bra{\texttt{-+}}$.
{This represents the effect of the rough merge for all measurement outcomes, as a CPTP map.}

A \textbf{smooth merge} performs the corresponding {operation, interchanging} the horizontal and vertical axes, and also $Z$ and $X$. {An} interstitial row of qubits are prepared in the $\ket{0}$ state, {and performing} $Z$-plaquette measurements across a horizontal join between two parent memories realises a measurement of $Z_L\ssur{1} Z_L\ssur{2}$ {with outcome} $M = \pm 1$.
The observables on the original and new memories satisfy $X_L\ssur{3} \equiv X_L\ssur{1} X_L\ssur{2}$, and either $Z_L\ssur{3} \equiv Z_L\ssur{1} \equiv Z_L\ssur{2}$ ({if~$M = +1$, i.e.~the positive branch}) or $Z_L\ssur{3} \equiv \pm Z_L\ssur{1} \equiv \mp Z_L\ssur{2}$ (if $M = -1$, where the signs {depend on the} choice of Pauli frame for the child memory).
{The effect on the logical state is} a CPTP map
$  \mathsf M_S(\rho) \,=\, K_{0,S}^{\phantom\dagger} \,\rho\, K_{0,S}^\dagger + K_{1,S}^{\phantom\dagger} \,\rho\, K_{1,S}^\dagger  $
where $K_{0,S}$ is realised if $M = +1$ and $K_{1,S}$ if $M = -1$. We have
$  K_{0,S} \;=\; \ket{0}\!\!\bra{00} \;+\; \ket{1}\!\!\bra{11}$
and {(depending on how we adapt the Pauli frame of the child memory)} one of the two choices
$  K_{1,S} \;=\; \ket{0}\!\!\bra{10} \;+\; \ket{1}\!\!\bra{01}$
or
$  K_{1,S} \;=\; \ket{0}\!\!\bra{01} \;+\; \ket{1}\!\!\bra{10}$.

The operations of {lattice surgery are very different} from those usually found in discussions of quantum computing. Describing split and merge operations {with standard circuits} is unwieldy, as they are explicitly non-unitary; {the merge operation in particular is intrinsically non-deterministic}.
As one may realise a \cnot by a smooth split of the control qubit, followed by a rough merge with {the target~\cite{melattice},} lattice surgery is usually presented as a way of realising operations in the unitary circuit model. However, a native language of splitting and merging for design, verification, and optimization of lattice surgery protocols would be 
valuable for the effective management of resources involving these operations~\cite{herrnoridevitt}.

%\vspace{-.5em}
\section{The ZX calculus}
%\vspace{-.25em}

The ZX calculus is a notation together with a system of transformations, for reasoning about tensors in terms of complementary bases~\cite{monster,zxbook}.
The ZX calculus can fruitfully be applied to certain processes in quantum information theory: specifically, the standard model of the calculus (using eigenstates of the Pauli $X$ and $Z$ operators) is effective for reasoning about stabilizer-like quantum operations~\cite{backens2014zx}, which are useful for describing transformations of Pauli observables.
In this section, we provide a brief introduction to the ``ZX calculus'' suitable for a non-specialist reader who is familiar with the quantum circuit model.
(A more complete introduction is available from Refs.~\cite{zxwebsite} and~\cite{zxbook}.)

\subsection{Using ZX diagrams to denote tensors}

The notation of the ZX calculus consists of graphs with coloured nodes --- which represent operations, including preparations and projections --- connected by edges,  representing qubits.
These graphs denote tensors in the same way that circuit diagrams do, and can often be interpreted as a sequence of linear operators acting on a state-vector.
(In the ZX diagrams of this article, the time axis runs left to right.)
In its simplest form, the ZX calculus includes nodes of only two colours --- conventionally red and green, where `red' is the darker shade --- representing two choices of basis in which information may be stored and transformed.
For example, preparations and measurements may be denoted by
%\vspace*{-.5ex}
\begin{equation}{}\!\!\!\!
\begin{aligned}
\begin{aligned}[t]
  \includegraphics[angle=-90]{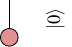}
\end{aligned}
\quad&
\begin{aligned}[t]
  \includegraphics[angle=-90]{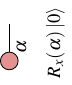}
\end{aligned}
&\qquad\;
\begin{aligned}[t]
  \includegraphics[angle=-90]{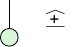}
\end{aligned}
\quad&
\begin{aligned}[t]
  \includegraphics[angle=-90]{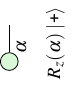}
\end{aligned}
&\qquad
\begin{aligned}[t]
  \includegraphics[angle=-90]{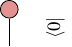}
\end{aligned}
\quad&
\begin{aligned}[t]
  \includegraphics[angle=-90]{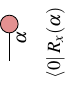}
\end{aligned}
&\qquad
\begin{aligned}[t]
  \includegraphics[angle=-90]{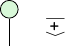}
\end{aligned}
\quad&
\begin{aligned}[t]
  \includegraphics[angle=-90]{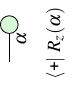}
\end{aligned}
\end{aligned}
\label{eqn:lollipops}\!\!\!\!
\end{equation}~\\%[-1ex]
where each diagram fragment represents the operator written below it.\footnote{%
  More precisely, the diagrams of Eqn.~\eqref{eqn:lollipops} actually represent operators $\sqrt 2\ket{0}$, $\sqrt 2\!\!\:\bra{0}$, $\sqrt 2\ket{\texttt+}$, $\sqrt 2\!\!\:\bra{\texttt+}$\,, and so forth: see~\emph{e.g.}~Ref.~\cite{Backens-2015}.
  This distinction is unimportant to our results, and the cumulated scalar $2^{k/2}$ for a diagram can be easily inferred from its topology.
}
(Setting $\alpha = \pi$ in each case yields the state or projector orthogonal to the corresponding unlabelled node; we omit angles as labels when they are multiples of $2\pi$.)
The other basic nodes of the ZX calculus are as follows:

%\vspace*{-1.5ex}
\begin{equation}
\begin{aligned}
\begin{aligned}[t]~\\[-2.5ex]
  \includegraphics[angle=-90]{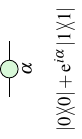}
\end{aligned}
\qquad&&\quad
\begin{aligned}[t]
  \includegraphics[angle=-90]{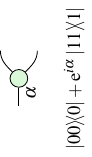}
\end{aligned}
\qquad&&\quad
\begin{aligned}[t]
  \includegraphics[angle=-90]{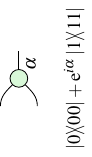}
\end{aligned}
\\[.5ex]
\begin{aligned}[t]~\\[-2.5ex]
  \includegraphics[angle=-90]{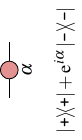}
\end{aligned}
\qquad&&\quad
\begin{aligned}[t]
  \includegraphics[angle=-90]{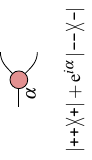}
\end{aligned}
\qquad&&\quad
\begin{aligned}[t]
  \includegraphics[angle=-90]{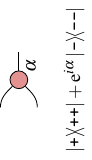}
\end{aligned}
\end{aligned}
\label{ZXnodes}
\end{equation}~\\%[-4ex]
Note that each of these nodes come in adjoint pairs, as follows (blank nodes denote either red or green, with the same colour throughout each equation):
%\vspace*{-1ex}%
\begin{equation}
\biggl[
\begin{aligned}~\\[-2.5ex]
\;\;
  \includegraphics[angle=-90]{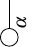}
\;\;
\\[-2.5ex]~
\end{aligned}
\biggr]^{\textstyle\dagger}
=\;\;
\begin{aligned}~\\[-2.5ex]
  \includegraphics[angle=-90]{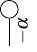}
\;
\\[-2.5ex]~
\end{aligned}
\;;
\qquad
\biggl[
\begin{aligned}~\\[-2.5ex]
\;\;
  \includegraphics[angle=-90]{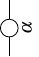}
\;\;
\\[-2.5ex]~
\end{aligned}
\biggr]^{\textstyle\dagger}
=\;\;
\begin{aligned}~\\[-2ex]
  \includegraphics[angle=-90]{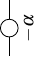}
\;
\\[-2.5ex]~
\end{aligned}
\;;
\qquad
\biggl[
\begin{aligned}~\\[-3.5ex]
\;\;
  \includegraphics[angle=-90]{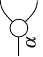}
\;\;
\\[-2.5ex]~
\end{aligned}
\biggr]^{\textstyle\dagger}
=\;\;
\begin{aligned}~\\[-3ex]
  \includegraphics[angle=-90]{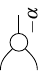}
\\[-2.5ex]~
\end{aligned}
\;.
\label{eqn:dagger}
\end{equation}~\\[-1ex]
We can generalise the nodes above to ones with larger in-degree and out-degree (conventionally known as ``spiders''~\cite{Coecke2008}), which for any given angle $\alpha$ corresponds to ${\ket{0}^{\!\otimes k} + \mathrm{e}^{i\alpha} \ket{1}^{\!\otimes k}}$ for the green nodes, and ${\ket{\texttt{+}}^{\!\otimes k} + \mathrm{e}^{i\alpha} \ket{\texttt{-}}^{\!\otimes k}}$ for the red nodes (taking the transpose of those tensor factors representing inputs rather than outputs).
Considering nodes {merely} as tensors, it is not important  whether a given wire represents a ``bra'' or a ``ket''; we may be agnostic about their direction, and even allow them to run vertically without ambiguity as to how to evaluate the tensor contraction with an index of {another tensor}.

Note that the left-hand diagrams of Eqn.~\eqref{ZXnodes} are $Z$ and $X$ rotations, respectively.
Then nodes of degree $2$ suffice to describe arbitrary single-qubit operations by their Euler decomposition.
To show how we may represent arbitrary unitary operators with this notation, it suffices to demonstrate a decomposition of a \cnot operator.
We may do this up to a scalar factor in the ZX notation as follows:
\vspace*{4ex}%
\begin{equation}
\begin{aligned}[b]
  \smash{
  \begin{aligned}
  ~\\[-4ex]
      \includegraphics[angle=-90]{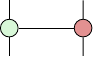}
    \\[-3ex]~ 
  \end{aligned}
  }
  	 \;\;\; {}={} \;\;&
	   \ket{0}\!\!\bra{0} \otimes \langle 0 \vert \texttt{+} \rangle \otimes \ket{\texttt{+}}\!\!\bra{\texttt{+}} \;+\;
	   \ket{0}\!\!\bra{0} \otimes \langle 0 \vert \texttt{-} \rangle \otimes \ket{\texttt{-}}\!\!\bra{\texttt{-}}
    \\&+\;
      \ket{1}\!\!\bra{1} \otimes \langle 1 \vert \texttt{+} \rangle \otimes\ket{\texttt{+}}\!\!\bra{\texttt{+}} \;+\;
	  \ket{1}\!\!\bra{1} \otimes \langle 1 \vert \texttt{-} \rangle \otimes \ket{\texttt{-}}\!\!\bra{\texttt{-}}
	\\[2ex]
	 {}={}\;&
	   \tfrac{1}{\sqrt 2} \ket{0}\!\!\bra{0} \otimes I \;+\;
	   \tfrac{1}{\sqrt 2} \ket{1}\!\!\bra{1} \otimes X
	 \;=\; \tfrac{1}{\sqrt 2}\,\mathrm{CNOT}.
  \end{aligned}
\label{cnot}
\end{equation}~\\[-1ex]
To represent \cnot exactly using ZX, we would multiply this by a scalar factor of $\sqrt 2$, \emph{e.g}~by including a diagram fragment which is equivalent to the scalar $\sqrt 2$.
Such scalar factors are often omitted as a minor abuse of notation when describing operators using the ZX calculus, but they play an important role for our results, in understanding such diagrams as representing Kraus operators.
In similar ways to \eqref{cnot} we may represent any unitary operator with the ZX calculus, using standard results in circuit decomposition.

\subsection{Properties of the ZX calculus}

As a tensor notation, the ZX calculus enjoys several convenient properties which correspond to abstract representations of equality tests and copying in either in the $Z$ or the $X$ eigenbases.
\begin{subequations}\allowdisplaybreaks
\label{eqn:bialgebras}%
These properties 
are as follows --- where throughout, blank nodes denote either red or green (with the same colour throughout each equation), provided that the sum of the angles in each diagram are equivalent mod $2\pi$:
\vspace*{0ex}\label{eq:frob}
\begin{gather}{}
\mspace{-18mu}
\begin{aligned}~\\[-5ex]
  \includegraphics[angle=-90]{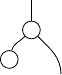}
\end{aligned}
\quad=\;\;\;
\begin{aligned}~\\[-5ex]
  \includegraphics[angle=-90]{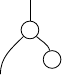}
\end{aligned}
\quad=\quad
\begin{aligned}~\\[-5ex]
  \includegraphics[angle=-90]{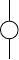}
\end{aligned}
\quad=\quad
\begin{aligned}~\\[-5ex]
  \includegraphics[angle=-90]{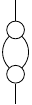}
\end{aligned}
\quad=\quad
\begin{aligned}~\\[-2.5ex]
  \includegraphics[angle=90]{ZXfigs/handbagBlank-left.pdf}
\end{aligned}
\;\;\;=\quad
\begin{aligned}~\\[-2.5ex]
  \includegraphics[angle=90]{ZXfigs/handbagBlank-right.pdf}
\end{aligned}
\;\;;
\label{eqn:Frobenius1-1}
\\[1ex]
\begin{aligned}~\\[-5ex]
  \includegraphics[angle=-90]{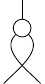}
\end{aligned}
\quad=\quad
\begin{aligned}~\\[-5ex]
  \includegraphics[angle=-90]{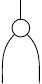}
\end{aligned}
\;\;;
\qquad
\begin{aligned}~\\[-2.5ex]
  \includegraphics[angle=90]{ZXfigs/mergeBlank-twist.pdf}
\end{aligned}
\quad=\quad
\begin{aligned}~\\[-2.5ex]
  \includegraphics[angle=90]{ZXfigs/mergeBlank-long.pdf}
\end{aligned}
\;\;;
\label{eqn:CommutativeLaws}
\\[1ex]
\begin{aligned}~\\[-6ex]
  \includegraphics[angle=-90]{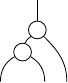}
\end{aligned}
\quad=\quad
\begin{aligned}~\\[-4ex]
  \includegraphics[angle=-90]{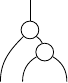}
\end{aligned}
\;\;;
\qquad
\begin{aligned}~\\[-2ex]
  \includegraphics[angle=90]{ZXfigs/assocBlank-left.pdf}
\end{aligned}
\quad=\quad
\begin{aligned}~\\[-3ex]
  \includegraphics[angle=90]{ZXfigs/assocBlank-right.pdf}
\end{aligned}
\;\;;
\label{eqn:AssociativeLaws}
\\[2ex]
\begin{aligned}~\\[-6.25ex]
  \includegraphics[angle=-90]{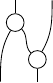}
\end{aligned}
\quad=\quad
\begin{aligned}~\\[-5ex]
  \includegraphics[angle=-90]{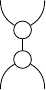}
\end{aligned}
\quad=\quad
\begin{aligned}~\\[-6.25ex]
  \includegraphics[angle=-90]{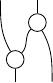}
\end{aligned}
\;\;.
\label{eqn:FrobeniusLaws}
\end{gather}~\\[-0ex]
\end{subequations}
These properties %\footnote{%
%  These properties are ones which one may consider for any bialgebra (of which the tensors generated by the red nodes, and separately by the green nodes, are examples).
%  We note that Eqn.~\eqref{eqn:Frobenius1-1} indicates that the bialgebras of the red nodes and the green nodes are both ``special'' bialgebras; Eqn.~\eqref{eqn:CommutativeLaws} indicate that they each are both commutative and co-commutative; Eqn.~\eqref{eqn:AssociativeLaws} indicates that they are associative and co-associative; and Eqn.~\eqref{eqn:FrobeniusLaws} indicates that they are both Frobenius bialgebras.
%  However, in practical work, it is not important to be aware of the names of these properties or to distinguish them from one another.
%}
allow us to reduce nodes of higher in- and out-degree in arbitrary ways, using the nodes of Eqns.~\eqref{eqn:lollipops} and~\eqref{ZXnodes}, so long as the composite diagram has the correct in- and out-degree and the same total angle~\cite[Ch.\ 8.6.1]{zxbook}.
For instance, using Eqns.~\eqref{eqn:AssociativeLaws} and~\eqref{eqn:FrobeniusLaws}, we may without confusion define nodes of degree 4 from compositions of nodes of degree 3 (where again we require that mod~$2\pi$ sum of the phases remains constant through each equation):
\vspace*{-1.5ex}
\begin{subequations}%
\begin{gather}{}
\mspace{-18mu}
\begin{aligned}~\\[-6ex]
  \includegraphics[scale=0.75,angle=-90]{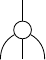}
\end{aligned}
\;:=\;
\begin{aligned}~\\[-6ex]
  \includegraphics[scale=0.75,angle=-90]{ZXfigs/assocBlank-left.pdf}
\end{aligned}
\;=\;
\begin{aligned}~\\[-4ex]
  \includegraphics[scale=0.75,angle=-90]{ZXfigs/assocBlank-right.pdf}
\end{aligned}
\,\,;
\\
\begin{aligned}~\\[-6ex]
  \includegraphics[scale=0.75,angle=-90]{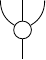}
\end{aligned}
\;:=\;
\begin{aligned}~\\[-2ex]
  \includegraphics[scale=0.75,angle=90]{ZXfigs/assocBlank-left.pdf}
\end{aligned}
\;=\;
\begin{aligned}~\\[-3ex]
  \includegraphics[scale=0.75,angle=90]{ZXfigs/assocBlank-right.pdf}
\end{aligned}
\,\,;
\\[1ex]
\begin{aligned}~\\[-6ex]
  \includegraphics[scale=0.75,angle=-90]{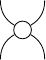}
\end{aligned}
\!\;:=\,
\begin{aligned}~\\[-6.25ex]
  \includegraphics[scale=0.75,angle=-90]{ZXfigs/zcurveBlank.pdf}
\end{aligned}
\,=\,
\begin{aligned}~\\[-5ex]
  \includegraphics[scale=0.75,angle=-90]{ZXfigs/bneckBlank.pdf}
\end{aligned}
\,=\,
\begin{aligned}~\\[-6.25ex]
  \includegraphics[scale=0.75,angle=-90]{ZXfigs/scurveBlank.pdf}
\end{aligned}
\,\,.\;\;
\mspace{-18mu}
\end{gather}~\\[-2ex]
\end{subequations}%
We may recursively define nodes of any degree in this way.
This is the so-called \emph{Spider {Law}}, {and also} allows us to simplify diagrams by merging nodes of degree $1$ and $2$ of the same colour by adding their angles~\cite[Ch.\ 9.4]{zxbook}. Eqn.~\eqref{eqn:Frobenius1-1} also allows for simplification by reducing the number of nodes: {for instance,} if the sum of the angles is a multiple of $2\pi$, the resulting operation simply represents the identity operator $\mathbbm 1$, in which case we may remove the node without changing the meaning of the diagram.

Finally, for angles $\alpha = k \pi$ for integer $k$, we may describe a simple rule for commuting
nodes of degree $2$, past differently coloured nodes of degree $3$:
%\vspace*{-2ex}
\begin{equation}
\begin{aligned}~\\[-5ex]
  \includegraphics[angle=-90]{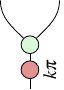}
\end{aligned}
\;\;=\;\;
\begin{aligned}~\\[-5ex]
  \includegraphics[angle=-90]{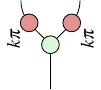}
\end{aligned}
\qquad\qquad\qquad
\begin{aligned}~\\[-5ex]
  \includegraphics[angle=-90]{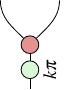}
\end{aligned}
\;\;=\;\;
\begin{aligned}~\\[-5ex]
  \includegraphics[angle=-90]{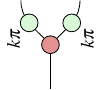}
\end{aligned}
\label{copy}
\end{equation}~\\[-2ex]%
The reflection of Eqn.~\eqref{copy} about the vertical axis also holds, and represents the sense in which the nodes of one colour correspond to simple arithmetic operations on information stored in the ``distinguished basis'' of the other colour.
These rules in effect provide a reduced instruction set for transformation and evaluation of tensors, which is well-suited for automated reasoning about equivalence of quantum procedures \cite{kissinger2015quantomatic}.

\subsection{A simple demonstration of how to use the ZX calculus}

We now provide simple illustrations of the way in which the ZX calculus can be used as a tool for computation, by demonstrating an elementary fact about the stabiliser formalism.
Our aim here is not to argue that the ZX calculus is in some way a superior tool to prove this result, but instead to \emph{demonstrate} how the ZX calculus produces this result --- and in so doing, demonstrate the transformations which we use in our results.
(For a more complete introduction to ZX, the interested reader is invited to refer to Refs.~\cite{zxwebsite} or~\cite{zxbook}.)

Our demonstration is to show how to derive how to commute Pauli operators past \cnot gates, using the ``Spider Rule'' (the higher-level result following from Eqns.~\eqref{eqn:bialgebras} allowing one to merge any two connected nodes of the same colour) together with Eqn.~\eqref{copy}.
The Spider Rule describes how we can merge, or split, two nodes of the same colour while keeping the total phase the same (mod~$2\pi$).
For example, with a node of degree $3$ and a node of degree $2$, where the node of degree $3$ happens not to have a phase of zero, we have:
\vspace*{-0.5ex}
\begin{equation}
\begin{alignedat}{2}{}\!\!\!\!
&\;
\begin{aligned}
\begin{aligned}[t]~\\[-4.5ex]
  \includegraphics[angle=-90]{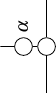}
\end{aligned}
\end{aligned}
\\[.75ex]
&\;\;\;\;\quad\Big\updownarrow
\\
\begin{aligned}
\begin{aligned}[t]~\\[-4.5ex]
  \includegraphics[angle=-90]{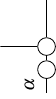}
\end{aligned}
\end{aligned}
\;\longleftrightarrow{}&\;
\begin{aligned}
\begin{aligned}[t]~\\[-4.5ex]
  \includegraphics[angle=-90]{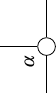}
\end{aligned}
\end{aligned}
&\longleftrightarrow{}&\;
\begin{aligned}
\begin{aligned}[t]~\\[-4.5ex]
  \includegraphics[angle=-90]{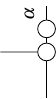}
\end{aligned}
\end{aligned}
~\\[-1.5ex]
\!\!\!\!
\label{phaseNodeCommute}
\end{alignedat}
\end{equation}%~\\[-2ex]
where the white nodes again stand for either red or green nodes (the same colour throughout).
This effectively allows us to ``commute'' a degree~2 node past a degree~3 node of the same colour.
Furthermore Eqn.~\eqref{copy} essentially describes how $X$ and $Z$ operators propagate between controls and targets of \cnot gates, as degree-2 nodes with a $\pi$ phase are either $X$ or $Z$ operators:
\vspace*{-2ex}
\begin{equation}
\begin{aligned}
\begin{aligned}[t]~\\[-2.5ex]
  \includegraphics[angle=-90]{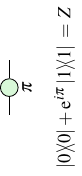}
\end{aligned}
\qquad&&\qquad
\begin{aligned}[t]~\\[-2.5ex]
  \includegraphics[angle=-90]{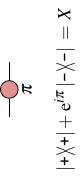}
\end{aligned}
\end{aligned}
\label{piNodes}
\end{equation}~\\[-1ex]
We may then verify, (\emph{i.e.},~compute from first principles) how to commute $Z$ and $X$ operators past a \cnot gate:
\begin{equation}
\begin{aligned}
\begin{aligned}
\begin{aligned}[t]~\\[-4.5ex]
  \includegraphics[angle=-90]{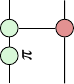}
\end{aligned}
\end{aligned}
\,&\to\,
\begin{aligned}
\begin{aligned}[t]~\\[-4.5ex]
  \includegraphics[angle=-90]{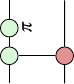}
\end{aligned}
\end{aligned}
\;;
&\qquad\quad
\begin{aligned}
\begin{aligned}[t]~\\[-4.5ex]
  \includegraphics[angle=-90]{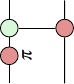}
\end{aligned}
\end{aligned}
\,&\to\,
\begin{aligned}
\begin{aligned}[t]~\\[-4.5ex]
  \includegraphics[angle=-90]{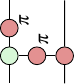}
\end{aligned}
\end{aligned}
\,\to\,
\begin{aligned}
\begin{aligned}[t]~\\[-4.5ex]
  \includegraphics[angle=-90]{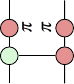}
\end{aligned}
\end{aligned}
\;;
\\[3ex]
\begin{aligned}
\begin{aligned}[t]~\\[-4.5ex]
  \includegraphics[angle=-90]{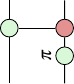}
\end{aligned}
\end{aligned}
\,&\to\,
\begin{aligned}
\begin{aligned}[t]~\\[-4.5ex]
  \includegraphics[angle=-90]{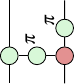}
\end{aligned}
\end{aligned}
\,\to\,
\begin{aligned}
\begin{aligned}[t]~\\[-4.5ex]
  \includegraphics[angle=-90]{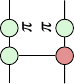}
\end{aligned}
\end{aligned}
\;;
&\qquad\quad
\begin{aligned}
\begin{aligned}[t]~\\[-4.5ex]
  \includegraphics[angle=-90]{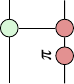}
\end{aligned}
\end{aligned}
\,&\to\,
\begin{aligned}
\begin{aligned}[t]~\\[-4.5ex]
  \includegraphics[angle=-90]{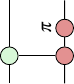}
\end{aligned}
\end{aligned}
\;.
\end{aligned}
\label{phaseCNOTcommute}
\end{equation}~\\[-1ex]
The equations (expressed as rewritings of diagrams) in the left-hand column  represent how $Z$ operators commute with the control of a \cnot, and propagate from the target to the control \cnot; the equations of the right-hand column represent how $X$ operators commute with the target of a \cnot, and propagate from the control to the target of a \cnot.
Thus, in the ZX calculus, the propagation of Pauli operators is a consequence of the rules for how the two colours of nodes interact, and can be demonstrated without matrix calculations.

%\vspace{-.5em}
\section{Lattice surgery in the ZX calculus}
%\vspace{-.25em}

Having in hand an account of lattice surgery using encoded Kraus operators, and given the presentation above of the tensors involved in the standard model of the ZX calculus, it should be clear that there are considerable similarities between them. A lattice ``split'' operation copies information represented in one of two bases, as shown in Section~\ref{sec:split}.
This mirrors the action {of Eqn.~\eqref{copy}}: for instance, a red $1$-to-$2$ node copies a green $\pi$-rotation, representing both how a change in the input maps to an output, and {also how} an input $Z$ observable is equivalent to a product of $Z$ observables at the output.
It has also been previously observed that the red node in the \cnot of Eqn.~\eqref{cnot} acts as an \xor \cite[Ch. 5.3.5]{zxbook}, which is also the action of a rough merge.
We thus appear to have an equivalence between red degree-$3$ nodes in the ZX calculus with rough lattice surgery operations, and between green degree-$3$ nodes and smooth operations.

We now make this equivalence precise and explicit, demonstrating that the actions of lattice surgery on encoded data essentially form a model of the ZX calculus.
The precise nature of the equivalence requires us to describe the merge operations in terms of \emph{ensembles} of simple diagrams, which yield byproduct operations associated with the heralded outcomes of the merge process.
The ZX calculus then provides us with a way to simplify the descriptions of these byproduct operations, allowing us to account for them in much more complicated procedures than the standard realisation of the \cnot gate. 

In what follows, a \emph{lattice surgery procedure} is a composition of the maps $\mathsf M_S$ and $\mathsf M_R$ (albeit possibly using different conventions for updating the Pauli frame in each instance) with $\mathsf S_S$ and $\mathsf S_R$, {and preparations of fresh qubits prepared in states $\ket{g_\alpha} \propto {\ket{0} + \mathrm{e}^{i\alpha} \ket{1}}$ and $\ket{r_\alpha} \propto {\ket{\texttt{+}} + \mathrm{e}^{i\alpha} \ket{\texttt{-}}}$.}

\begin{lemma}
The effect on the logical state space of a surface code of a smooth split, and of the positive branch of a smooth merge, are given by the operators associated with the green degree-3 nodes in Eqn.~\eqref{ZXnodes} for $\alpha = 0$; and similarly for the rough split/merge operations and the red degree-3 nodes.
\end{lemma}
\begin{proof}
  For the split operations, this follows from the equality of the operators associated to the {$1$-to-$2$} nodes in Eqn.~\eqref{ZXnodes} for $\alpha = 0$, and the corresponding Kraus operators $U_S$ and $U_R$ (respectively) of the unitary embeddings $\mathsf S_S$ and $\mathsf S_R$ of Section~\ref{split}.
  Similarly, the Kraus operators $K_{0,S}$ and $K_{0,R}$ (respectively) of the merge operators $\mathsf M_S$ and $\mathsf M_R$ of Section \ref{merge} are equal to the operators associates to (respectively) the green and the red {$2$-to-$1$} nodes in Eqn.~\eqref{ZXnodes}.
\end{proof}
\noindent
The preceding observation about identical pairs of linear operators, has the following consequence on how we can reason about compositions of operations in lattice surgery:
\begin{corollary}
  \label{cor:blimeyZXpositivelyIsLatticeSurgery}
  The positive branches of lattice surgery procedures (i.e.~conditioned on $+1$ outcomes of all merge operations) provide a model for {the equational theory of the ZX calculus} without projections, in which nodes of degree $2$ or $3$ have an angle of $0$.
\end{corollary}
\noindent
In particular: if we let
  \begin{minipage}{1.25em}\hspace*{-.75ex}\rotatebox{-90}{%
    $\begin{aligned}\begin{tikzpicture}[scale=0.25]
      \node [draw=black,fill=white,circle,inner sep=0.3ex] (v) at (0,0) {};
      \draw [out=135,in=-90] (v) to ++(-2.5ex,5ex);
      \draw [out=45,in=-90] (v) to ++(2.5ex,5ex);
      \draw (v) to ++(0,-5ex);
    \end{tikzpicture}\end{aligned}$}%
  \end{minipage}\
  denote the action of the unitary embedding $U_S$ (respectively, $U_R$) on encoded data realised by a smooth (resp.~a rough) split,
  \begin{minipage}{1.25em}\hspace*{-.75ex}\rotatebox{-90}{%
    $\begin{aligned}\begin{tikzpicture}[scale=-0.25]
      \node [draw=black,fill=white,circle,inner sep=0.3ex] (v) at (0,0) {};
      \draw [out=135,in=-90] (v) to ++(-2.5ex,5ex);
      \draw [out=45,in=-90] (v) to ++(2.5ex,5ex);
      \draw (v) to ++(0,-5ex);
    \end{tikzpicture}\end{aligned}$}%
  \;\;\end{minipage}\
  denote the action of the Kraus operator $K_{0,S}$ (resp.~$K_{0,R}$) on encoded data realised in the positive branch by a smooth (resp.~a rough) merge,
    \smash{\begin{minipage}{1.25em}\hspace*{-.75ex}\raisebox{5ex}{\rotatebox{-90}{%
      $\begin{aligned}\begin{tikzpicture}[scale=0.25]
        \node [draw=none,circle,inner sep=0.3ex,
              outer sep=-2pt, label=right:\rotatebox{90}{\footnotesize$\alpha$}] (v) at (0,0) {};
        \draw (v) to ++(0,5ex);
        \draw (v) to ++(0,-5ex);
        \filldraw [draw=black,fill=white] (v) circle (1.75ex);
        \end{tikzpicture}\!\end{aligned}$}}%
  \end{minipage}}
  denote an $R_x(\alpha)$ gate (resp.~a $R_z(\alpha)$ gate), and
  \smash{\begin{minipage}{1.25em}\hspace*{-1ex}\raisebox{5ex}{\rotatebox{-90}{%
      $\begin{aligned}\begin{tikzpicture}[scale=0.25]
        \node [draw=none,circle,inner sep=0.3ex,
              outer sep=-2pt, label=right:\rotatebox{90}{\footnotesize$\alpha$}] (v) at (0,0) {};
        \draw (v) to ++(0,5ex);
        \filldraw [draw=black,fill=white] (v) circle (1.75ex);
        \end{tikzpicture}\!\end{aligned}$}}%
  \end{minipage}}
  denote (up to an extra factor of $\sqrt 2$) preparation of a $\ket{g_\alpha}$ state (resp.~a $\ket{r_\alpha}$ state), then these operations satisfy all of the properties of Eqns.~\eqref{eqn:dagger} and~\eqref{eqn:bialgebras}.
More complicated green or red nodes (including nodes of degree $>1$ with non-zero angles) can then be realised using the {Spider Law~\cite[Ch.\ 9.4]{zxbook}}, describing them as the effect of compositions of split and merge operations in the positive branch.

The qualification ``in the positive branch'' in Corollary~\ref{cor:blimeyZXpositivelyIsLatticeSurgery} is significant, and relates to the normalisation of the diagram in Eqn.~\eqref{cnot}.
Omitting the correction operations needed to realise \cnot deterministically, the standard lattice-surgical realisation of {\cnot~\cite{melattice}} is by the composition $(\mathbf{1} \otimes \mathsf M_R)(\mathsf S_S \otimes \mathbf{1})$.
We may associate to this composition a pair of Kraus operators $(I \otimes K_{0,R})(U_S \otimes I)$ and $(I \otimes K_{1,R})(U_S \otimes I)$, where the positive branch corresponds to the first of these.
A simple calculation reveals that
$  (I \otimes K_{0,R})(U_S \otimes I)  = \tfrac{1}{\sqrt 2} \,\mathrm{CNOT}$,
which reflects the fact that the positive branch occurs with probability $\tfrac{1}{2}$ and realises a \cnot operation on all input states.
From this standpoint, the subnormalisation of Eqn.~\eqref{cnot} is a feature {of ZX notation, not} a bug: it captures not only the way in which states transform in the positive branch of a lattice-surgery procedure, but also the $2$-norms of the Kraus operators which govern the transformation.

This motivates a view of simple ZX diagrams as denoting Kraus maps of lattice surgery procedures.
It remains to describe how to represent the {``negative branch''} of rough and smooth merges.
\begin{lemma}
  \label{lem:negativeBranchMaps}%
  \begin{subequations}%
  For a rough merge, if we adapt the Pauli frame of the child memory to agree with the second parent memory so that $K_{1,R} = \ket{\texttt{\upshape+}}\!\!\bra{\texttt{\upshape-+}} + \ket{\texttt{\upshape-}}\!\!\bra{\texttt{\upshape+-}}$ \;--- \textup{respectively:} if we adapt the Pauli frame of the child memory to agree with the first parent memory so that $K_{1,R} = \ket{\texttt{\upshape+}}\!\!\bra{\texttt{\upshape+-}} + \ket{\texttt{\upshape-}}\!\!\bra{\texttt{\upshape-+}}$ \;---
  then
  \vspace*{0ex}%
  \begin{equation}
  K_{1,R} \;\;=\;\;
  \begin{aligned}~\\[-7ex]
    \includegraphics[angle=-90]{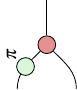}
  \end{aligned}
  \qquad
  \biggl(\text{respectively,}\quad   K_{1,R} \;\;=\;\;
  \begin{aligned}~\\[-4ex]
    \includegraphics[angle=-90]{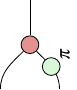}
  \end{aligned}
  \;\biggr).
  \end{equation}~\\[-1ex]
  Similarly, for a smooth merge, if we adapt the Pauli frame of the child memory to agree with the second parent memory so that $K_{1,S} = \ket{0}\!\!\bra{10} + \ket{1}\!\!\bra{01}$ \;--- \textup{respectively:} if we adapt the Pauli fram of the child memory to agree with the first parent memory so that $K_{1,S} = \ket{0}\!\!\bra{01} + \ket{1}\!\!\bra{10}$ \;---
  then
  \vspace*{0ex}%
  \begin{equation}
  K_{1,S} \;\;=\;\;
  \begin{aligned}~\\[-7ex]
    \includegraphics[angle=-90]{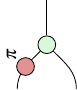}
  \end{aligned}
  \qquad
  \biggl(\text{respectively,}\quad   K_{1,S} \;\;=\;\;
  \begin{aligned}~\\[-3.5ex]
    \includegraphics[angle=-90]{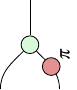}
  \end{aligned}
  \;\biggr).
  \end{equation}%
  \end{subequations}%
\end{lemma}%
\begin{proof}
  By calculation of the operators, from Eqn.~\eqref{ZXnodes}.  
\end{proof}

\begin{corollary}%
  \label{cor:ZXdiagramsMergeKrausOptors}%
  \begin{subequations}%
  \label{eqn:mergeDiagrams}%
  A rough merge operation realises a logical transformation $\mathsf M_R$ with Kraus operators
  \vspace*{-0ex}
  \begin{equation}
  \label{eqn:roughMergeDiagrams}
  \Biggl\{\;\;\;
    \begin{aligned}~\\[-5ex]
    \includegraphics[angle=-90]{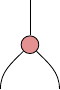}
    \end{aligned}
  \;\;\text{\Large\raisebox{-1.8ex},}\;\;\;
    \begin{aligned}~\\[-7ex]
    \includegraphics[angle=-90]{ZXfigs/mergeX-lcorr.pdf}
    \end{aligned}
  \;\;\Biggr\}
  \qquad
  \text{or}
  \qquad
  \Biggl\{\;\;\;
    \begin{aligned}~\\[-5ex]
    \includegraphics[angle=-90]{ZXfigs/mergeX-pos.pdf}
    \end{aligned}
  \;\;\text{\Large\raisebox{-1.8ex},}\;\;\;
    \begin{aligned}~\\[-3ex]
    \includegraphics[angle=-90]{ZXfigs/mergeX-rcorr.pdf}
    \end{aligned}
  \;\;\Biggr\}
  \end{equation}~\\[-1ex]
  acting on the logical qubit of the surface codes, depending on whether one adopts the convention of correcting either the first or the second parent in the negative branch.
  Similarly, a smooth merge operation realises a logical transformation $\mathsf M_S$ with Kraus operators
  \vspace*{-0ex}
  \begin{equation}
  \label{eqn:smoothMergeDiagrams}
  \Biggl\{\;\;\;
    \begin{aligned}~\\[-5ex]
    \includegraphics[angle=-90]{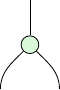}
    \end{aligned}
  \;\;\text{\Large\raisebox{-1.8ex},}\;\;\;
    \begin{aligned}~\\[-7ex]
    \includegraphics[angle=-90]{ZXfigs/mergeZ-lcorr.pdf}
    \end{aligned}
  \;\;\Biggr\}
  \qquad
  \text{or}
  \qquad
  \Biggl\{\;\;\;
    \begin{aligned}~\\[-5ex]
    \includegraphics[angle=-90]{ZXfigs/mergeZ-pos.pdf}
    \end{aligned}
  \;\;\text{\Large\raisebox{-1.8ex},}\;\;\;
    \begin{aligned}~\\[-3ex]
    \includegraphics[angle=-90]{ZXfigs/mergeZ-rcorr.pdf}
    \end{aligned}
  \;\Biggr\}
  \end{equation}~\\[-1ex]
  acting on the logical qubit of the surface codes, depending on whether on adopts the convention of correcting either the first or the second parent in the negative branch.
  \end{subequations}
\end{corollary}
\begin{proof}
  This follows from the previous Lemmata concerning the equality of the Kraus operators.  
\end{proof}

%\vspace*{-3ex}
\paragraph{Remark.}\upshape
In practise, it is likely to prove convenient to analyse the operators $\mathsf M_S$ and $\mathsf M_R$ with \emph{annotated ZX diagrams}~\cite{DP-2010}, in which we allow nodes whose phases are not necessarily constants, and may include variables whose values may only be determined during the computation (\emph{e.g.}~as with the results of a measurement).
\begin{subequations}%
For instance, we may denote the pairs of logical Kraus operators for $\mathsf M_R$ by
 % \vspace*{-1ex}
  \begin{equation}
  \label{eqn:roughMergeDiagrams}
  \begin{aligned}~\\[-7ex]
    \includegraphics[angle=-90]{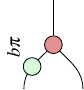}
  \end{aligned}
  \qquad\qquad
  \text{or}
  \qquad\qquad
  \begin{aligned}~\\[-3ex]
    \includegraphics[angle=-90]{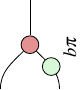}
  \end{aligned}
  \end{equation}~\\[-2ex]
  depending on the conventional choice of how one adapts the Pauli frame of the child memory; here $b \in \{0,1\}$ indicates the measurement outcome $M = (-1)^b$ of that merge, and the Kraus operator $K_{b,R}$ which is realised as a transformation of the quantum state.
  Similarly, the pairs of logical Kraus operators of $\mathsf M_S$ may be denoted by
 % \vspace*{-1.5ex}
  \begin{equation}
  \label{eqn:smoothMergeDiagrams}
  \begin{aligned}~\\[-7ex]
    \includegraphics[angle=-90]{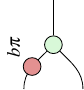}
  \end{aligned}
  \qquad\qquad
  \text{or}
  \qquad\qquad
  \begin{aligned}~\\[-3ex]
    \includegraphics[angle=-90]{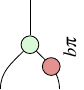}
  \end{aligned}
  \end{equation}~\\[-1.5ex]
  depending on the conventional choice of how one adapts the Pauli frame of the child memory.
  In either case, the bit $b \in \{0,1\}$ is revealed when the operation is performed, as a result of the stabiliser measurements which are performed during the merge operations; one may reason about the operations performed in general by leaving $b$ as a formal indeterminate, and describing the phase by a polynomial with real coefficients over some set of indeterminates (associated operationally with the outcomes of operations).
  \end{subequations}

\medskip
The purpose of associating the operations of lattice surgery with simple (annotated) ZX diagrams is two-fold: (a)~to allow us to reason about more diverse compositions of lattice surgery operations using compositions of the diagrams, and (b)~to allow us to relate more complex ZX-diagrams to lattice surgery procedures which {may be used to} realise them.
For (b), it suffices to produce a simple ZX diagram which is equivalent, and then consider how this diagram may be realised using (a). 
We have already glimpsed how (a) might be done in Corollary~\ref{cor:blimeyZXpositivelyIsLatticeSurgery} for the positive branch of any lattice surgery procedure; the following extends this to arbitrary outcomes of the merges:

\medskip
\begin{theorem}
  \label{thm:latticeSurgeryZXgeneral}
  Lattice surgery procedures model randomly-constructed simple ZX diagrams, in which gadgets realising $2$-to-$1$ nodes are selected from either one of the pairs of operations of Eqn.~\eqref{eqn:roughMergeDiagrams} for rough merges, and from either one of the pairs of operations of Eqn.~\eqref{eqn:smoothMergeDiagrams} for smooth merges; and the probability of each diagram is given by its 2-norm.
\end{theorem}
\begin{proof}
  Each pair of diagrams from Eqns.~\eqref{eqn:mergeDiagrams} describes Kraus operators of $\mathsf M_S$ and $\mathsf M_R$ (as determined by a conventional choice of procedure to update the Pauli frame).
  A single gadget selected from one of these pairs is modelled by a lattice merge of the appropriate sort, for which the diagram represents a Kraus operator.
  Composing such maps (and others) to produce a simple ZX diagram represents the effect of composing the Kraus operators for the lattice merge operations with those of {split operations, and preparation of qubits in the $\ket{g_\alpha}$ and $\ket{r_\alpha}$ states}. These form Kraus operators for the entire lattice surgery procedure.
  The 2-norm of the diagram then corresponds to the probability of the Kraus operator (i.e.~the particular merge outcomes) being realised.
\end{proof}

%\vspace{-1em}
\section{Consequences of the equivalence}
%\vspace{-.25em}

We have described a tight connection between the ZX calculus and lattice surgery. Splitting and merging can be written in terms of degree-3 nodes (including correction operations if necessary); realised as Kraus operators of lattice surgery procedures. 
The use of ZX as a tool for manipulating lattice surgery diagrams has a huge number of potential applications. To begin with, almost all work on the ZX calculus can now be applied directly to lattice surgery. In particular, known equivalences under re-writing in the calculus now have an interpretation in terms of equivalence of lattice surgery patterns. For example, the Frobenius laws of Eqn.~\eqref{eq:frob} can now be imported and read directly as referring to equivalent set of split and merge operations. The standard re-write axioms of the calculus give further pattern equivalences. To give a flavour of the power of this re-writing system, we here detail two examples where ZX diagrammatic equivalence gives novel lattice surgery procedures: for non-Clifford rotations, and the \cnot (as a subroutine).

%\vspace{-.3em}
\subsection{Lattice surgery T gate}
\label{sec:magicState}
%\vspace{-.2em}

Magic states {permit} the realisation of {logical operations outside of the Clifford group} on surface codes, allowing for (approximately) universal quantum computing~\cite{bravyimagic}. In the surface code, they may be  distilled to high fidelity, injected into a code surface, and then used to perform teleported rotation gates \cite{fowler2009high}, as these cannot be performed transversally \cite{PhysRevLett.102.110502}. A $T$ gate uses a magic state and teleportation to perform a $\pi/4$ rotation. The usual presentation of the teleported $T$ gate involves a \cnot and measurement (see e.g.~\cite{campbell2017roads}). However, the ZX calculus allows us to give a more efficient procedure for lattice surgery.

Consider the state $\ket{A} \propto {\ket{0} + \mathrm e^{i\pi/4}\ket{1}}$. In a teleported $T$ gate this is used to perform a $Z$-rotation $R_z(\pi/4) = \ket{0}\!\!\bra{0} + \mathrm{e}^{i\pi/4} \ket{1}\!\!\bra{1}$. In the ZX calculus the $\ket{A}$ state is represented by a green $\pi/4$ {preparation node}, and the $R_z(\pi/4)$ operation by a green $\pi/4$ {degree-2} node.
As an application of Eqn.~\eqref{eqn:bialgebras}, we have
\vspace*{4ex}
\begin{equation}
\label{eqn:positiveBranchTmerge}
\begin{aligned}~\\[-10ex]
  \includegraphics[angle=-90]{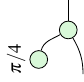}
\end{aligned}
\quad=\quad
\begin{aligned}~\\[-8ex]
  \includegraphics[angle=-90]{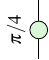}
\end{aligned}\;.
\end{equation}~\\[-2ex]
By Corollary~\ref{cor:blimeyZXpositivelyIsLatticeSurgery}, the left-hand diagram represents the positive branch of a smooth merge of an {$\ket{A}$} state with an input state. %, which suggests this as a possible to realise magic-state injection.
Using Lemma~\ref{lem:negativeBranchMaps} and Eqn.~\eqref{eqn:bialgebras}, the negative branch of this procedure is described by
\vspace*{2ex}
\begin{equation}
\label{eqn:negativeBranchTmerge-a}
\begin{aligned}~\\[-11ex]
  \includegraphics[angle=-90]{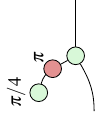}
\end{aligned}
\quad=\quad
\begin{aligned}~\\[-8ex]
  \includegraphics[angle=-90]{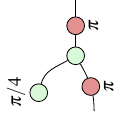}
\end{aligned}
\quad=\quad
\begin{aligned}~\\[-6ex]
  \includegraphics[angle=-90]{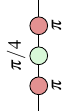}
\end{aligned}
\quad=\quad
\begin{aligned}~\\[-8ex]
  \includegraphics[angle=-90]{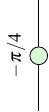}
\end{aligned}\;.
\end{equation}~\\[-3ex]
As with a standard $T$ gate~\cite[\S 6]{fowler2009high}, the rotation to correct this to $+\pi/4$ uses {$\ket{Y} \propto {\ket{0} + \mathrm e^{i\pi/2}\ket{1}}$}:
\vspace*{3.5ex}
\begin{equation}
\begin{aligned}~\\[-10ex]
  \includegraphics[angle=-90]{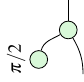}
\end{aligned}
\quad=\quad
\begin{aligned}~\\[-8ex]
  \includegraphics[angle=-90]{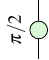}
\end{aligned}\;.
\end{equation}~\\[-2.5ex]
The negative branch of this operation is:
\vspace*{2ex}
\begin{equation}
\label{eqn:negativeBranchSmerge}
\begin{aligned}~\\[-11ex]
  \includegraphics[angle=-90]{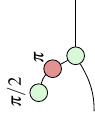}
\end{aligned}
\quad=\quad
\begin{aligned}~\\[-8ex]
  \includegraphics[angle=-90]{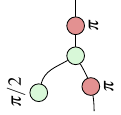}
\end{aligned}
\quad=\quad
\begin{aligned}~\\[-6ex]
  \includegraphics[angle=-90]{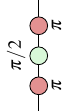}
\end{aligned}
\quad=\quad
\begin{aligned}~\\[-8ex]
  \includegraphics[angle=-90]{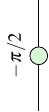}
\end{aligned}\;.
\end{equation}~\\[-2.5ex]
This can now be corrected, if necessary, by a phase flip. The $T$ gate has now become simple merging with magic states.\footnote{%
 The alternative convention for the negative branch of merges works equally well (see Section~\ref{apx:TgateMerge} of the Appendix).
}
{Note that the above ZX analysis also} holds if we interchange red and green {nodes: we thus also obtain a} procedure for $R_x(\pi/4)$ rotations as well. Procedures for smaller-angle rotations can similarly be found, which can then be used for the most efficient magic state compilations~\cite{campbell2016efficient}.

\vspace{-.25em}
\subsection{Realising \cnot gates}
\vspace{-.15em}

The original lattice surgery procedure for the \cnot gate is a smooth split of the control, followed by a rough merge of one of the daughter memories with the target \cite[\S4.1]{melattice}, followed by appropriate corrections depending on the outcome of the measurements in the merge procedure.
Compare the positive branch of this procedure (using Corollary~\ref{cor:blimeyZXpositivelyIsLatticeSurgery}) to the (subnormalised) representation of \cnot from Eqn.~\eqref{cnot}:
\vspace*{-0.5ex}
\begin{equation}
\label{eqn:standardSurgeryCNOT}
\begin{aligned}~\\[-5ex]
  \includegraphics[angle=-90]{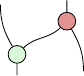}
\end{aligned}  
\;\;\;=\;\;\;
\begin{aligned}~\\[-5ex]
  \includegraphics[angle=-90]{ZXfigs/basicCNOT.pdf}
\end{aligned}  
\end{equation}~\\[-0.5ex]
These two presentations are clearly topologically equivalent. By considering other topologically equivalent presentations of the coloured graph of Eqn.~\eqref{cnot}, and how those graphs describe operations which may be decomposed as ZX diagrams, we may obtain  a further set of procedures, all of which implement the \cnot in the positive branch (in the Appendix we show how we may realise \cnot operations in the negative branch for each of these 
%\pagebreak
\noindent procedures, as well as for that given in {Eqn.~\eqref{eqn:standardSurgeryCNOT}}):
\vspace*{0ex}
\begin{equation}
\label{eqn:standardSurgeryCNOT:procedure2}
\begin{aligned}~\\[-5ex]
  \includegraphics[angle=-90]{ZXfigs/basicCNOT.pdf}
\end{aligned}  
\;\;\;=\;\;\;
\begin{aligned}~\\[-5ex]
  \includegraphics[angle=-90]{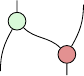}
\end{aligned}  
\;\;=\;\;\;
\begin{aligned}~\\[-5ex]
  \includegraphics[angle=-90]{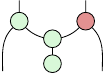}
\end{aligned}  
\;\;\;=\;\;\;
\begin{aligned}~\\[-5ex]
  \includegraphics[angle=-90]{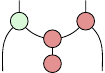}
\end{aligned}  
\;\;\;=\;\;\;
\begin{aligned}~\\[-5ex]
  \includegraphics[angle=-90]{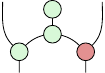}
\end{aligned}  
\;\;\;=\;\;\;
\begin{aligned}~\\[-5ex]
  \includegraphics[angle=-90]{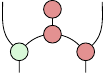}
\end{aligned}  
\end{equation}~\\[-.51ex]
From left to right: the second of these diagrams describes a rough split of the target followed by a smooth merge with the control; the third and fourth show the creation of an intermediary Bell pair $\ket{\Phi^+} \propto {\ket{00} + \ket{11}} = {\ket{\texttt{++}} + \ket{\texttt{--}}}$, the two halves of which are merged with the control and target respectively; and the final two
describe splits to both control and target qubits (smooth and rough respectively), with a simulated Bell projection $\bra{\Phi^+} \propto \bra{00} + \bra{11} = \bra{\texttt{++}} + \bra{\texttt{--}}$ on a pair of the daughter memories.

Thus, by minor variations of ZX diagrams, we obtain a proliferation of different procedures to realise a \cnot by lattice surgery procedures --- of which only the one illustrated in Eqn.~\eqref{eqn:standardSurgeryCNOT} has previously appeared in the literature.
This gives an indication of the breadth of compilation flexibility that the use of the ZX calculus will bring to the surface code with lattice surgery.

%\vspace{-.7em}
\section{Conclusions}
%\vspace{-.3em}

We have demonstrated how the ZX calculus acts as a precise and fundamental description of the splitting and merging operations of lattice surgery on the planar surface code. Merge operations output a bit of information, and this is used to determine which out of a set of possible diagrams describe the post-merge state. Sometimes, as with the case of the constructed \cnot gate, these possible diagrams differ by logical Pauli corrections. However, this is not in general the case. Previously there was no easy, systematic way to describe these Pauli frame updates in the high-level language used (\emph{i.e.},~the circuit model). With the use of the ZX calculus, there is now a high-level notation, together with rules to calculate with that notation.
This allows one to straightforwardly describe the Pauli frame information associated with a lattice surgery procedure, as well as how that information propagates under further operations.

Lattice surgery fits precisely the structure of the ZX calculus in the ``positive branch'', \emph{i.e.}, when merge measurement outcomes are $+1$. Merge and split operations --- whether rough (corresponding to red nodes) or smooth (green) --- are related through the dagger structure of the calculus. A rough merge is a retraction (inversion by post-composition) of a rough split, and equivalently for smooth merges and splits. Legs of nodes, which represent different parent or child surfaces, may be interchanged; this is the commutative structure. The specialness axioms also follow, as does associativity. Because of this, multi-input/output ``spiders'' may be defined. In lattice surgery terms, this is a simultaneous splitting or merging into/from multiple surfaces. Multi-surface splitting was described in the original paper on lattice surgery \cite{melattice}; now, however, we can prove what effect such operations have on the encoded data, and how they interact with each other.
In particular, the ability to describe the byproducts in the ``negative branches'', and when they can be corrected, will make it possible to determine which multi-node operations can be deterministically performed.

We have focussed in this paper on lattice surgery within the \textit{surface} code.
The question of a similar representation for other codes is an interesting question.
In colour code lattice surgery, for example, $Y$-basis merging appears \cite{ccode} (see also \cite{litinski2019game}): this may motivate the development of a calculus closely related to the ZX calculus in which such merges have a similarly direct representation.
In other ways, ZX already has proven its worth as a language for the broad class of codes known as coherent parity check codes (which includes the CSS construction) \cite{chancellor16}; and an early use of automated theorem proving in the calculus was verify correctness of the Steane code \cite{duncan2013verifying}.
Ways in which ZX may be used as a notation to describe the use of stabilizer codes is an important area of ongoing further research.

{There are a number of uses for the work presented here; the following is by no means exhaustive. Firstly, an immediate application is the use of ZX as the basis of an \emph{intermediate representation} in compilers} for near-term error corrected quantum devices. Such devices, based for example on architectures such as \cite{nickerson2014freely,nickerson2013topological}, have strictly limited resources. Compiling protocols to \cnot operations {wastes valuable operations; using ZX we can now compile to precisely the operations that such a device will in practice implement.}
The use of ZX, and automated tools \cite{pyzx}, for compilation is an ongoing larger research project, including optimisations for Clifford circuits \cite{FaganDuncan}, circuit simplification \cite{duncan2019graph} including routing and topological constraints \cite{kissinger2019cnot,cowtan2019qubit}, and state-of-the-art T-count reductions at the time of writing \cite{kissinger2019tcount}.
{Secondly, the ZX calculus can be used as a quantum protocol design tool. Through lattice surgery, the ZX calculus expresses another circuit-like model of quantum computing which we may attempt to realise through practical operations.} {The ZX calculus allows us to explore new techniques for quantum protocols, which may not be as naturally expressed by unitary circuits (such as in measurement-based quantum computing \cite{duncan2009graph}), for computation and also networking/communication~\cite{jones2016design,van2014} and large-scale qubit registers \cite{scalable}} There is a large amount of work on the calculus that can now be imported for use with lattice surgery; the new procedures given in this paper will be the first of many examples.
With the equivalence between ZX and lattice surgery fixed, it will now be possible to take a diagram in the ZX calculus, re-write it, and then interpret that diagram as a lattice surgery procedure --- which will include a definition of determinism in diagrams. The equivalence shown in this paper allows us to analyse lattice surgery procedures, to determine how a ZX transformation may practically be realised. This will include the question of time-ordering (including simultaneity) of compiled operations.

There are further uses of the ZX/lattice surgery equivalence. By producing ZX-compiled lattice surgery procedures for known algorithms and protocols, not only can we potentially simplify procedures significantly, but also a different logic can be applied to analyse how these protocols behave. ZX from its beginnings was introduced as another logic to analyse and design further quantum algorithms and protocols. By expanding the space of design tools, we expand what can be thought of in quantum computing. With a physical model of the ZX calculus in the operations of lattice surgery, this programme 
has an important role to play in realising near-term error corrected quantum devices.

\section*{Acknowledgements}
Many thanks to Ross Duncan and Aleks Kissinger for many useful and interesting discussions on the topic of this paper, and to anonymous reviewers at QPL for comments on an earlier version.

NB is supported by the EPSRC National Hub in Networked Quantum Information Technologies (NQIT.org).
{DH was supported by EPSRC under Grant EP/L022303/1.}\\

\bibliography{bib-ZX-lattice}

\vspace{2cm}
\appendix
\renewcommand\thesubsection{\Alph{subsection}.}
\renewcommand\thesubsubsection{\thesubsection\arabic{subsubsection}.}

\vspace{-.5em}
\section*{Appendix}
\vspace{-.25em}

\input{appendix}

\end{document}

%% file: appendix.tex
In this Appendix we show how the correction operations for the \cnot procedures of Eqns.~\eqref{eqn:standardSurgeryCNOT} and~\eqref{eqn:standardSurgeryCNOT:procedure2} equate to logical Pauli corrections on the control and/or target qubit, and also show how the choice of convention for adapting merge operations in $T$-state injection does not introduce any essential difficulties.

\subsection{{CNOT patterns}}

Eqns.~\eqref{eqn:standardSurgeryCNOT} and~\eqref{eqn:standardSurgeryCNOT:procedure2} describe transformations of logical qubits which are realised in the ``positive branch'' of lattice surgery procedures involving various merges, splits, preparations, and projections, which may be inferred from the topology of the diagram involved.
Below, we describe other possible evolutions which those procedures may realise, for different conventional choices of the Kraus operations $K_{1,S}$ or $K_{1,R}$ for the ``negative'' branch of the merge operations involved.

In each case, the transformations realised differ from the transformations realised in the positive branch only by logical Pauli operations (represented by green and/or red $\pi$-nodes).
These logical Pauli operators are heralded by the outcome of the merge operations, and can therefore be tracked or corrected.
This demonstrates that these protocols are not sensitive to the convention used for adapting the Pauli frame of the child memory in any given merge operation.

\vspace*{-2ex}
\subsubsection{Smooth split control, rough merge with target. }

\vspace*{-1ex}
\begin{subequations}%
\allowdisplaybreaks
\begin{align}
  \begin{minipage}{40mm}\raggedright
    Positive branch:    
  \end{minipage}
&&%\mspace{-40mu}
  \begin{gathered}~\\[-5ex]
    \includegraphics[scale=0.8,angle=-90]{ZXfigs/scurveCNOT.pdf}
  \end{gathered}\qquad&
\mspace{80mu}\\[4ex]
  \begin{minipage}{40mm}\raggedright
      Negative branch, using \\
      $K_{1,R} \,=\, \ket{\texttt{+}}\!\!\bra{\texttt{-+}} \;+\; \ket{\texttt{-}}\!\!\bra{\texttt{+-}}$\,:
  \end{minipage}
&&%\mspace{-40mu}
  \begin{gathered}~\\[-5ex]
    \includegraphics[scale=0.8,angle=-90]{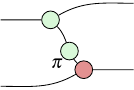}
  \end{gathered}
  \quad&=\quad
  \begin{gathered}~\\[-5ex]
    \includegraphics[scale=0.8,angle=-90]{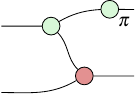}  
  \end{gathered}
\mspace{80mu}\\[4ex]
  \begin{minipage}{40mm}\raggedright
      Negative branch, using \\
      $K_{1,R} \,=\, \ket{\texttt{+}}\!\!\bra{\texttt{+-}} \;+\; \ket{\texttt{-}}\!\!\bra{\texttt{-+}}$\,:
  \end{minipage}
&&\mspace{-80mu}
  \begin{gathered}~\\[-5ex]
    \includegraphics[scale=0.8,angle=-90]{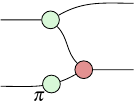}
  \end{gathered}
  \quad&=\quad
  \begin{gathered}~\\[-5ex]
    \includegraphics[scale=0.8,angle=-90]{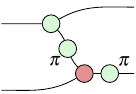}
  \end{gathered}
  \quad=\quad
  \begin{gathered}~\\[-5ex]
    \includegraphics[scale=0.8,angle=-90]{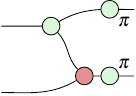}
  \end{gathered}
\end{align}  
\end{subequations}

\vspace*{-2ex}
\subsubsection{Rough split target, smooth merge with control.}

\vspace*{-1ex}
\begin{subequations}%
\allowdisplaybreaks
\begin{align}
  \begin{minipage}{40mm}\raggedright
    Positive branch:    
  \end{minipage}
&&%\mspace{-80mu}
  \begin{gathered}~\\[-5ex]
    \includegraphics[scale=0.8,angle=-90]{ZXfigs/zcurveCNOT.pdf}
  \end{gathered}\qquad&
\\[4ex]
  \begin{minipage}{40mm}\raggedright
      Negative branch, using \\
      $K_{1,S} \,=\, \ket{0}\!\!\bra{10} \;+\; \ket{1}\!\!\bra{01}$\,:
  \end{minipage}
&&%\mspace{-80mu}
  \begin{gathered}~\\[-5ex]
    \includegraphics[scale=0.8,angle=-90]{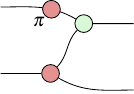}
  \end{gathered}
  \quad&=\quad
  \begin{gathered}~\\[-5ex]
    \includegraphics[scale=0.8,angle=-90]{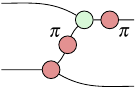}
  \end{gathered}
  \quad=\quad
  \begin{gathered}~\\[-5ex]
    \includegraphics[scale=0.8,angle=-90]{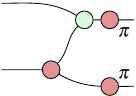}
  \end{gathered}
\\[4ex]
  \begin{minipage}{40mm}\raggedright
      Negative branch, using \\
      $K_{1,S} \,=\, \ket{0}\!\!\bra{01} \;+\; \ket{1}\!\!\bra{10}$\,:
  \end{minipage}
&&\mspace{-80mu}
  \begin{gathered}~\\[-5ex]
    \includegraphics[scale=0.8,angle=-90]{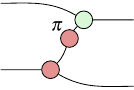}
  \end{gathered}
  \quad&=\quad
  \begin{gathered}~\\[-5ex]
    \includegraphics[scale=0.8,angle=-90]{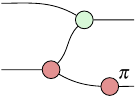}  
  \end{gathered}
\end{align}
\end{subequations}

\vspace*{-2ex}
\subsubsection{Create intermediary Bell pair then merge with control and target.}

By preparing a single-qubit state in the $\ket{0}$ state and then performing a smooth split, or preparing a single-qubit $\ket{\texttt+}$  state and then performing a rough split, we may create a Bell pair, expressed alternatively as $\tfrac{1}{\sqrt 2} \bigl( \ket{00} + \ket{11} \bigr)$ or $\tfrac{1}{\sqrt 2} \bigl(\ket{\texttt{++}} + \ket{\texttt{--}}\bigr)$.
We may represent this state by a curved wire producing two outputs, or ``cup'' --- representing the fact that the coefficients of this tensor are the same as that of the identity matrix --- if we are willing to account for the fact that this representation is super-normalised by a factor of $\sqrt 2$ (as with the representations of $\ket{0}$ and $\ket{\texttt+}$ described in Eqn.~\eqref{eqn:lollipops}).
We may then use this to realise a CNOT operation (with some probability) by performing a smooth and a rough merge.
These merge operations may independently yield the transformation for the positive branch (realising a Kraus operator $K_{0,\ast}$) or the negative branch (realising a Kraus operator $K_{1,\ast}$).

In the following, we illustrate only the analysis in which \emph{both} merges yield the positive branch, or the negative branch. (The cases where one merge yields the positive branch and one the negative branch are easier, but otherwise similar.)
\begin{subequations}%
\allowdisplaybreaks%
Unlike for merge operations, no logical byproduct operations are produced by lattice surgery split operations: and so it does not matter from that point of view whether the ``cup'' is created using a red split or a green split:
\vspace*{0ex}
\begin{align}{}
\mspace{-24mu}
  \begin{minipage}{40mm}\raggedright
    Positive branch:    
  \end{minipage}
&&\mspace{0mu}
  \begin{gathered}~\\[-5ex]
  \includegraphics[scale=0.8,angle=-90]{ZXfigs/prep-zbell-CNOT.pdf}
\end{gathered}  
\quad\!&=\quad
  \begin{gathered}~\\[-5ex]
  \includegraphics[scale=0.8,angle=-90]{ZXfigs/prep-xbell-CNOT.pdf}
\end{gathered}  
\quad\!=\;\;\;
  \begin{gathered}~\\[-5ex]
  \includegraphics[scale=0.8,angle=-90]{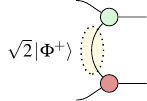}
\end{gathered}
\\[-4ex]
\notag
\intertext{
We therefore treat both these cases together, in terms of ``cups'' as in the right-hand-side above rather than explicit preparations and splits.
In the following, we sometimes commute $X$ and $Z$ operators (represented by `dark' red / `light' green $\pi$-phase nodes) past one another: this induces an unimportant global phase of $-1$ when we do so (which we denote by $\cong$ rather than $=$ for clarity).
}  
\mspace{-24mu}
  \begin{minipage}{40mm}\raggedright
      Negative branch, using \\
      $K_{1,S} \,=\, \ket{0}\!\!\bra{10} \;+\; \ket{1}\!\!\bra{01}$\,,
      $K_{1,R} \,=\, \ket{\texttt{+}}\!\!\bra{\texttt{-+}} \;+\; \ket{\texttt{-}}\!\!\bra{\texttt{+-}}$\,:
  \end{minipage}
&&\mspace{0mu}
  \begin{gathered}~\\[-5ex]
    \includegraphics[scale=0.8,angle=-90]{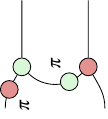}
  \end{gathered}
  \;\;&=\;\;
  \begin{gathered}~\\[-5ex]
    \includegraphics[scale=0.8,angle=-90]{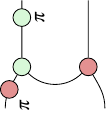}
  \end{gathered}
  \;\;=\;\;
  \begin{gathered}~\\[-5ex]
    \includegraphics[scale=0.8,angle=-90]{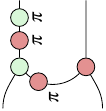}
  \end{gathered}
  \;\;=\;\;
  \begin{gathered}~\\[-5ex]
    \includegraphics[scale=0.8,angle=-90]{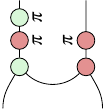}
  \end{gathered}
  \mspace{-20mu}
\\[4ex]
\mspace{-24mu}
  \begin{minipage}{40mm}\raggedright
      Negative branch, using \\
      $K_{1,S} \,=\, \ket{0}\!\!\bra{10} \;+\; \ket{1}\!\!\bra{01}$\,,
      $K_{1,R} \,=\, \ket{\texttt{+}}\!\!\bra{\texttt{-+}} \;+\; \ket{\texttt{-}}\!\!\bra{\texttt{+-}}$\,:
  \end{minipage}
&&\mspace{0mu}
  \begin{gathered}~\\[-5ex]
    \includegraphics[scale=0.8,angle=-90]{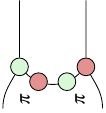}
  \end{gathered}
  \;\;&\cong\;\;
  \begin{gathered}~\\[-5ex]
    \includegraphics[scale=0.8,angle=-90]{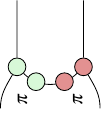}
  \end{gathered}
  \;\;=\;\;
  \begin{gathered}~\\[-5ex]
    \includegraphics[scale=0.8,angle=-90]{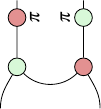}
  \end{gathered}
  \mspace{-20mu}
\\[4ex]
\mspace{-24mu}
  \begin{minipage}{40mm}\raggedright
      Negative branch, using \\
      $K_{1,S} \,=\, \ket{0}\!\!\bra{01} \;+\; \ket{1}\!\!\bra{10}$\,,
      $K_{1,R} \,=\, \ket{\texttt{+}}\!\!\bra{\texttt{+-}} \;+\; \ket{\texttt{-}}\!\!\bra{\texttt{-+}}$\,:
  \end{minipage}
&&\mspace{0mu}
  \begin{gathered}~\\[-5ex]
    \includegraphics[scale=0.8,angle=-90]{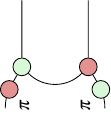}
  \end{gathered}
  \;\;&=\;\;
  \begin{gathered}~\\[-5ex]
    \includegraphics[scale=0.8,angle=-90]{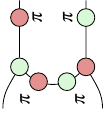}
  \end{gathered}
  \;\;\cong\;\;
  \begin{gathered}~\\[-5ex]
    \includegraphics[scale=0.8,angle=-90]{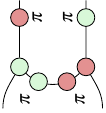}
  \end{gathered}
  \;\;=\;\;
  \begin{gathered}~\\[-5ex]
    \includegraphics[scale=0.8,angle=-90]{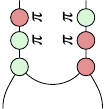}
  \end{gathered}
  \mspace{-20mu}
\\[4ex]
\mspace{-24mu}
  \begin{minipage}{40mm}\raggedright
      Negative branch, using \\
      $K_{1,S} \,=\, \ket{0}\!\!\bra{01} \;+\; \ket{1}\!\!\bra{10}$\,,
      $K_{1,R} \,=\, \ket{\texttt{+}}\!\!\bra{\texttt{+-}} \;+\; \ket{\texttt{-}}\!\!\bra{\texttt{-+}}$\,:
  \end{minipage}
&&\mspace{0mu}
  \begin{gathered}~\\[-5ex]
    \includegraphics[scale=0.8,angle=-90]{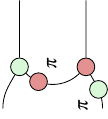}
  \end{gathered}
  \;\;&=\;\;
  \begin{gathered}~\\[-5ex]
    \includegraphics[scale=0.8,angle=-90]{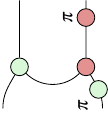}
  \end{gathered}
  \;\;=\;\;
  \begin{gathered}~\\[-5ex]
    \includegraphics[scale=0.8,angle=-90]{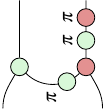}
  \end{gathered}
  \;\;=\;\;
  \begin{gathered}~\\[-5ex]
    \includegraphics[scale=0.8,angle=-90]{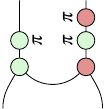}
  \end{gathered}
  \mspace{-20mu}
\end{align}
\end{subequations}

\vspace*{-2ex}
\subsubsection{Split control and target, and simulate an intermediary Bell measurement.}
\vspace*{-1ex}

The adjoint of the above diagrams --- involving preparing the state $\ket{\Phi^+}$ and then merging the two halves into the control (smooth merge) and target (rough merge) --- corresponds to a procedure in which we perform a smooth split operation on the control, a rough split operation on the target, and then perform a Bell measurement involving a child memory of each of these splits to obtain the outcome $\ket{\Phi^+}$.
(The ``projector'' $\sqrt 2 \bra{\Phi^+}$ is represented by a ``cap'', or a curved wire with two inputs, again representing the fact that as a tensor it has the same coefficients as the identity matrix.) 
As we cannot perform this measurement deterministically, we may simulate it using one of two different merge procedures:
\vspace*{0ex}
\begin{equation}
  \begin{gathered}~\\[-5ex]
  \includegraphics[scale=0.8,angle=-90]{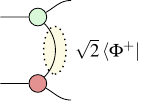}
\end{gathered}  
\;\;=\;\;
  \begin{gathered}~\\[-5ex]
  \includegraphics[scale=0.8,angle=-90]{ZXfigs/proj-zbell-CNOT.pdf}
\end{gathered}  
\quad=\;\;
  \begin{gathered}~\\[-5ex]
  \includegraphics[scale=0.8,angle=-90]{ZXfigs/proj-xbell-CNOT.pdf}
\end{gathered}
\end{equation}~\\[-1ex]
In each case, both the merge and the measurement yields non-deterministic results.
The measurement nodes illustrated here (representing $\bra{\texttt{+}}$ or $\bra{0}$ respectively, up to scalar factors) both have a phase of $0$; the alternative outcome in each case is a phase of $\pi$ (representing $\bra{\texttt{-}}$ or $\bra{1}$, respectively).
We may represent the possible outcomes of the measurement using a phase $b \pi$, depending on a bit $b \in \{0,1\}$ which indicates the result of the measurement.
This will allow us to reduce the possible transformations to the positive branch, regardless of the outcome of the measurement.

\vspace*{-2ex}
\paragraph{(a) Bell measurement using a smooth merge}

\begin{subequations}%
\allowdisplaybreaks%
\begin{align}{}
\mspace{-24mu}
  \begin{minipage}{40mm}\raggedright
    Positive branch \\ of merge:
  \end{minipage}
&&%\mspace{-24mu}
  \begin{gathered}~\\[-5ex]
  \includegraphics[scale=0.8,angle=-90]{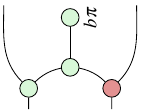}
\end{gathered}  
\;\;&=\;\;
  \begin{gathered}~\\[-5ex]
  \includegraphics[scale=0.8,angle=-90]{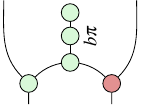}
\end{gathered}  
\;\;=\;\;
  \begin{gathered}~\\[-5ex]
  \includegraphics[scale=0.8,angle=-90]{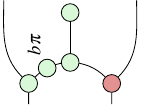}
\end{gathered}  
\;\;=\;\;
  \begin{gathered}~\\[-5ex]
  \includegraphics[scale=0.8,angle=-90]{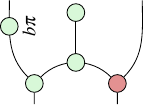}
\end{gathered}  
\mspace{-20mu}
\\[4ex]
\mspace{-24mu}
  \begin{minipage}{40mm}\raggedright
    Negative branch \\ of merge, using
      $K_{1,S} \,=\, \ket{0}\!\!\bra{10} \,+\, \ket{1}\!\!\bra{01}$\,:
  \end{minipage}
&&%\mspace{-36mu}
  \begin{gathered}~\\[-5ex]
    \includegraphics[scale=0.8,angle=-90]{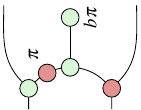}
  \end{gathered}  
  \;\;&=\;\;
  \begin{gathered}~\\[-5ex]
    \includegraphics[scale=0.8,angle=-90]{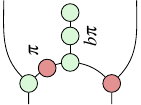}
  \end{gathered}  
  \;\;\cong\;\;
  \begin{gathered}~\\[-5ex]
    \includegraphics[scale=0.8,angle=-90]{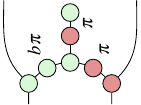}
  \end{gathered}  
  \;\;=\;\;
  \begin{gathered}~\\[-5ex]
    \includegraphics[scale=0.8,angle=-90]{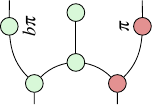}
  \end{gathered}  
  \mspace{-20mu}
\intertext{%
(Here, $\cong$ denotes equality up to an irrelevant phase factor of $(-1)^b$.
Note that we simply absorb the red $\pi$-phase into the green degree-1 node, corresponding to the equation $\bra{\texttt{+}} = \bra{\texttt{+}} X$ of operators.)
}
\mspace{-24mu}
  \begin{minipage}{40mm}\raggedright
    Negative branch \\ of merge, using
      $K_{1,S} \,=\, \ket{0}\!\!\bra{01} \,+\, \ket{1}\!\!\bra{10}$\,:
  \end{minipage}
&&%\mspace{-36mu}
  \begin{gathered}~\\[-5ex]
    \includegraphics[scale=0.8,angle=-90]{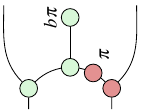}
  \end{gathered}  
  \;\;&=\;\;
  \begin{gathered}~\\[-5ex]
    \includegraphics[scale=0.8,angle=-90]{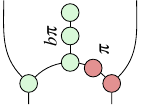}
  \end{gathered}  
  \;\;=\;\;
  \begin{gathered}~\\[-5ex]
    \includegraphics[scale=0.8,angle=-90]{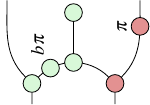}
  \end{gathered}  
  \;\;=\;\;
  \begin{gathered}~\\[-5ex]
    \includegraphics[scale=0.8,angle=-90]{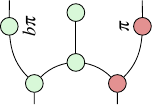}
  \end{gathered}  
  \mspace{-20mu}
\end{align}
\end{subequations}

\vspace*{-2ex}
\paragraph{(b) Bell measurement using a rough merge}

\begin{subequations}%
\allowdisplaybreaks%
\begin{align}{}
\mspace{-24mu}
  \begin{minipage}{40mm}\raggedright
    Positive branch \\ of merge:
  \end{minipage}
&&%\mspace{-18mu}
  \begin{gathered}~\\[-5ex]
  \includegraphics[scale=0.8,angle=-90]{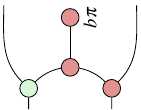}
\end{gathered}  
\;\;&=\;\;
  \begin{gathered}~\\[-5ex]
  \includegraphics[scale=0.8,angle=-90]{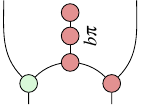}
\end{gathered}  
\;\;=\;\;
  \begin{gathered}~\\[-5ex]
  \includegraphics[scale=0.8,angle=-90]{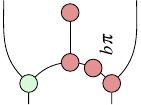}
\end{gathered}  
\;\;=\;\;
  \begin{gathered}~\\[-5ex]
  \includegraphics[scale=0.8,angle=-90]{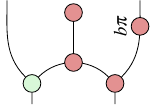}
\end{gathered}  
\mspace{-20mu}
\\[4ex]
\mspace{-24mu}
  \begin{minipage}{40mm}\raggedright
    Negative branch \\ of merge, using
      $K_{1,R} \,=\, \ket{\texttt{+}}\!\!\bra{\texttt{-+}} \,+\, \ket{\texttt{-}}\!\!\bra{\texttt{+-}}$\,:
  \end{minipage}
&&%\mspace{-36mu}
  \begin{gathered}~\\[-5ex]
    \includegraphics[scale=0.8,angle=-90]{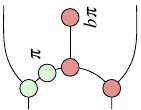}
  \end{gathered}  
  \;\;&=\;\;
  \begin{gathered}~\\[-5ex]
    \includegraphics[scale=0.8,angle=-90]{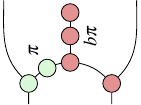}
  \end{gathered}  
  \;\;=\;\;
  \begin{gathered}~\\[-5ex]
    \includegraphics[scale=0.8,angle=-90]{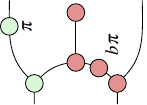}
  \end{gathered}  
  \;\;=\;\;
  \begin{gathered}~\\[-5ex]
    \includegraphics[scale=0.8,angle=-90]{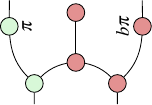}
  \end{gathered}  
  \mspace{-20mu}
\\[4ex]
\mspace{-24mu}
  \begin{minipage}{40mm}\raggedright
    Negative branch \\ of merge, using
      $K_{1,R} \,=\, \ket{\texttt{+}}\!\!\bra{\texttt{+-}} \,+\, \ket{\texttt{-}}\!\!\bra{\texttt{-+}}$\,:
  \end{minipage}
&&%\mspace{-36mu}
  \begin{gathered}~\\[-5ex]
    \includegraphics[scale=0.8,angle=-90]{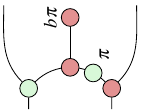}
  \end{gathered}  
  \;\;&=\;\;
  \begin{gathered}~\\[-5ex]
    \includegraphics[scale=0.8,angle=-90]{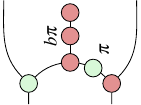}
  \end{gathered}  
  \;\;\cong\;\;
  \begin{gathered}~\\[-5ex]
    \includegraphics[scale=0.8,angle=-90]{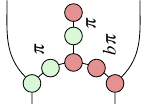}
  \end{gathered}  
  \;\;=\;\;
  \begin{gathered}~\\[-5ex]
    \includegraphics[scale=0.8,angle=-90]{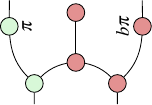}
  \end{gathered}  
  \mspace{-20mu}
\end{align}
(Here again, $\cong$ denotes equality up to an irrelevant phase factor of $(-1)^b$.
Note that we simply absorb the red $\pi$-phase into the green projection node, corresponding to the equation $\bra{0} = \bra{0} Z$ of operators.)
\end{subequations}

\subsection{Magic state merge patterns}
\label{apx:TgateMerge}

In Section~\ref{sec:magicState}, we demonstrate how to realise a $T$ rotation by merging a magic state with an input qubit.
In this Section, we demonstrate that --- similarly to the realisations of \cnot\ above --- the choice of convention for the negative branch of the merge does not introduce any essential difficulties in the interpretation of $T$ gate teleportation as being realised by merge operations.

The effect of the positive branch of the merge operation is the same as illustrated in Eqn.~\eqref{eqn:positiveBranchTmerge} regardless of the convention for the negative branch; we show that the opposite convention to (the left-hand side of) Eqn.~\eqref{eqn:negativeBranchTmerge-a} for the negative merge is equivalent.
In that alternative convention, the red $\pi$-node appears on the opposite input to the $2$-to-$1$ node than the magic state, yielding the transformation
\vspace*{2ex}
\begin{equation}
\label{eqn:negativeBranchTmerge-b}
\begin{aligned}~\\[-8ex]
  \includegraphics[scale=0.8,angle=-90]{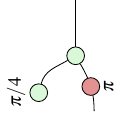}
\end{aligned}
\;=\;
\begin{aligned}~\\[-6ex]
  \includegraphics[scale=0.8,angle=-90]{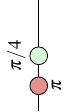}
\end{aligned}
\;=\;
\begin{aligned}~\\[-6ex]
  \includegraphics[scale=0.8,angle=-90]{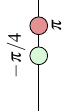}
\end{aligned}\;.
\end{equation}~\\[-2ex]
This operator differs from the operation on the right-hand side of Eqn.~\eqref{eqn:negativeBranchTmerge-a} by a red $\pi$ node, which is an $X$ operation; the two procedures then differ only by a $\nott$ operation controlled on the merge outcome.

Notice that the effect of the negative branch described by Eqn.~\eqref{eqn:negativeBranchTmerge-b} corresponds to the $T$ gate teleportation procedure of \cite[Fig.~10.25]{nandc}, but without realising this operation in terms of a \cnot\ gate in the form of a decomposition such as in Eqn.~\eqref{cnot}.
This also demonstrates a simple example of the usefulness of the ZX calculus to simplify quantum information processing procedures, as follows.
If we express the effect of the negative branch shown in Eqn.~\eqref{eqn:negativeBranchTmerge-a} in a form similar to \cite[Fig.~10.25]{nandc}, the role of the control and target are swapped. (The target in the new circuit is still the qubit which is measured.) Interestingly however, it omits the classically-controlled $\nott$ correction.
By the interchangeability of the two conventions for adapting the Pauli frame in a merge operation, the ZX calculus thus allows us to arrive at a protocol that is simpler to realise --- even for a well-established procedure in quantum information processing.